\documentclass[11pt]{article}

\usepackage[utf8]{inputenc}
\usepackage[english]{babel}
\usepackage{amsmath}
\usepackage{nicefrac}
\usepackage{amsthm}
\usepackage{amsfonts}
\usepackage{verbatim}
\usepackage{color}
\usepackage[arrow, matrix, curve]{xy}
\usepackage{bbm}
\usepackage{eqparbox}
\usepackage[numbers]{natbib}
\usepackage{stmaryrd}
\usepackage{amssymb}
\usepackage{mathrsfs}
\usepackage{pictexwd,dcpic}
\usepackage{mathabx}

\DeclareMathSymbol{\leq}{\mathrel}{symbols}{20}
   
\DeclareMathSymbol{\geq}{\mathrel}{symbols}{21}

\newtheoremstyle{WreschTheoremstyle} % Name
                        {1.5em}    % Space above
                        {2.5em}    % Space below
                        {}         % Body font
                        {}         % Indent amount
                        {\bfseries}% Theorem head font
                        {}        % Punctuation after theorem head
                        {\newline} % Space after theorem head
                        {\raisebox{0.6em}{\thmname{#1}\thmnumber{#2}\thmnote{ (#3)}}}% Theorem head spec (can be left empty, meaning 'normal')

\newcommand{\R}{\mathbb{R}}

\newcommand{\N}{\mathbb{N}}

\newcommand{\K}{\mathcal{K}}

\newcommand{\E}{\mathbb{E}}
\newcommand{\e}{\varepsilon}
\newcommand{\Lb}{\mathcal{L}}
\renewcommand{\1}{\mathbbm{1}}
\newcommand{\esssup}{\mathrm{ess}\sup}

\newtheorem{Theorem}{Theorem}[section]
\newtheorem{Proposition}[Theorem]{Proposition}

\newtheorem{Lemma}[Theorem]{Lemma}
\newtheorem{Remark}[Theorem]{Remark}
\newtheorem{Definition}[Theorem]{Definition}
\newtheorem{Example}[Theorem]{Example}

\numberwithin{equation}{section}

\makeatletter
\newcommand{\customlabel}[1]{%
     \stepcounter{ref}%
   \protected@write
\@auxout{}{\string\newlabel{#1}{{\thesatz.\arabic{ref}}{\thepage}{\thesatz.\arabic{ref}}{#1}{}}}%
   \hypertarget{#1}{\thesatz.\arabic{ref}}%
}
\makeatother

\newenvironment{sciabstract}{\begin{quote}}{\end{quote}}

\newcounter{lastnote}

\title{Stochastic averaging for a spatial population model in random environment}
\newcommand{\pdftitle} {Stochastic averaging for a spatial population model in random environment}
\newcommand{\pdfauthor}{Martin Friesen}

\author{
Martin Friesen\footnote{Department of Mathematics, Wuppertal University, Germany, friesen@math.uni-wuppertal.de}\\
Yuri Kondratiev\footnote{Department of Mathematics, Bielefeld University, Germany, kondrat@math.uni-bielefeld.de}
}

\usepackage[plainpages=false,pdfpagelabels=true,bookmarks=true,pdfauthor={\pdfauthor},
pdftitle={\pdftitle}]{hyperref}

\makeatletter
\def\HyPsd@CatcodeWarning#1{}
\makeatother

\begin{document}


%\baselineskip24pt


\maketitle

\begin{sciabstract}\textbf{Abstract:}
 In this work we study the non-equilibrium Markov state evolution for a spatial population model on the space of locally finite configurations 
 $\Gamma^2 = \Gamma^+ \times \Gamma^-$ over $\R^d$ where particles are marked by spins $\pm$. 
 Particles of type '+' reproduce themselves independently of each other and, moreover,
 die due to competition either among particles of the same type or particles of different type.
 Particles of type '-' evolve according to a non-equilibrium Glauber-type dynamics with activity $z$ and potential $\psi$.
 Let $L^S$ be the Markov operator for '+' -particles and $L^E$ the Markov operator for '-' -particles.
 The non-equilibrium state evolution $(\mu_t^{\e})_{t \geq 0}$ is obtained from the Fokker-Planck equation 
 with Markov operator $L^S + \frac{1}{\e}L^E$, $\e > 0$, which itself is studied in terms of correlation function evolution on a suitable chosen scale of Banach spaces.
 We prove that in the limiting regime $\e \to 0$ the state evolution $\mu_t^{\e}$ converges weakly to some state evolution $\overline{\mu}_t$
 associated to the Fokker-Planck equation with (heuristic) 
 Markov operator obtained from $L^S$ by averaging the interactions of the system with the environment with respect to the 
 unique invariant Gibbs measure of the environment.
\end{sciabstract}


\section{Introduction}

\subsection{General introduction}
The mathematical theory of complex systems related with models of spatial ecology, biology and genetics
is a fast developing area in modern mathematics providing many challenging tasks \cite{K97, N01, SEW05}.
In this work we consider the particular class of birth-and-death models where particles may die and create new particles according to some prescribed 
birth-and-death rates. Models without spatial structure have been studied since the early works of Kato and Kolmogorov (see e.g. \cite{KATO54}).
A recent account on related results is given in \cite{BLM06}.
Motivated by applications (see e.g. \cite{FFHKKK15}), a particular branch of modern probability theory 
is devoted to the study spatial models with a finite or infinite number of individuals.
The case of finite population dynamics was considered, e.g., in \cite{EW01, KOL06, F16b} (see also the references therein).
Birth-and-death dynamics for infinite population models are more subtle and require a much more detailed analysis (see \cite{GK06, FKK12, FK16b}).
Due to the mathematical complexity of the models, it is necessary to study each model separately.
 
\subsection{The model}
In this work we study a particular two-type birth-and-death model on the state space of locally finite configurations $\Gamma^2 = \Gamma^+ \times\Gamma^-$
over $\R^d$, where
\[
 \Gamma^{\pm} = \left \{ \gamma^{\pm} \subset \R^d \ | \ |\gamma^{\pm} \cap \Lambda| < \infty \ \text{ for all compacts } \Lambda \subset \R^d \right\}.
\]
Here $|\gamma^{\pm} \cap \Lambda|$ denotes the number of points in the set $\gamma^{\pm} \cap \Lambda$. 
For simplicity of notation we let $\gamma^{\pm} \cup x$, $\gamma^{\pm} \backslash x$ stand for $\gamma^{\pm} \cup \{x\}$, $\gamma^{\pm} \backslash \{x\}$
and write $\gamma = (\gamma^+,\gamma^-) \in \Gamma^2$.
The first component $\gamma^+$ describes the microscopic configuration of the system whereas $\gamma^-$ the microscopic configuration of the environment. 

Dynamics of the environment is described by state evolution on $\Gamma^-$ with (heuristic) Markov operator of Glauber-type
acting on functions $F \in \mathcal{FP}(\Gamma^-)$ (specified in the next section)
\begin{align}\label{EQ:31}
 (L^EF)(\gamma^-) &= \sum \limits_{x \in \gamma^-}(F(\gamma^- \backslash x) - F(\gamma^-)) + z \int \limits_{\R^d}e^{-E_{\psi}(x,\gamma^-)}(F(\gamma^- \cup x) - F(\gamma^-))d x.
\end{align}
Here $z > 0$ is the activity and $\psi \geq 0$ is a pair potential associated to the relative energy
\[
 E_{\psi}(x, \gamma^-) = \sum \limits_{y \in \gamma^-}\psi(x-y) \in [0,\infty], \ \ x \in \R^d, \ \ \gamma^- \in \Gamma^-.
\]
Such an environment has been studied e.g. in \cite{FKK12a, FKK12b}.

Dynamics of the system is described by state evolution on $\Gamma^+$ with birth-and-death rates depending on the configuration $\gamma^-$.
More precisely we study an infinite population spatial logistic model (short BDLP model) with (heuristic) Markov operator
\begin{align}\label{EQ:32}
 (L^SF)(\gamma) &= \sum \limits_{x \in \gamma^+}\left( m + \sum \limits_{y \in \gamma^+ \backslash x}a^-(x-y) + g\sum \limits_{y \in \gamma^-}b^-(x-y)\right)(F(\gamma^+ \backslash x, \gamma^-) - F(\gamma))
 \\ \notag &\ \ \ + \sum \limits_{x \in \gamma^+}\int \limits_{\R^d}a^+(x-y)(F(\gamma^+ \cup y, \gamma^-) - F(\gamma^+, \gamma^-))d y,
\end{align}
where $F \in \mathcal{FP}(\Gamma^2)$ belongs to a class of functions specified in the next section.
Here $m \geq 0$ is the constant mortality rate, $a^- \geq 0$ the dispersion kernel and $a^+ \geq 0$ the branching kernel.
The difference to the BDLP model is due to additional interactions of the system with the environment via interaction kernel $b^- \geq 0$
and coupling constant $g \geq 0$. The original BDLP model ($g = 0$) was introduced by Bolker, Dieckmann, Law, Pacala in \cite{BP97, BP99, MDL04}.
Its mathematical properties has been studied e.g. in \cite{FM04, FKK09, KK16}.

\subsection{Assumptions}
The following are our main assumptions for this work.
\begin{enumerate}
\item[(E)] $z > 0$ and $\psi \geq 0$ is symmetric with $1 - e^{-\psi} \in L^1(\R^d)$. For $\alpha^- \in \R$ let
 \[
  C_{\psi}(\alpha^-) := \exp\left( e^{\alpha^-}\int \limits_{\R^d}\left( 1 - e^{-\psi(y)}\right)d y \right).
 \]
 Assume that there exists $\alpha_*^- \in \R$ such that 
 \begin{align}\label{EQ:00}
  z e^{-\alpha_*^-}C_{\psi}(\alpha_*^-) < 1.
 \end{align}
 \item[(S)] $m \geq 0$, $a^+ \neq 0$, $0 \leq a^{\pm}, b^- \in L^1(\R^d) \cap L^{\infty}(\R^d)$ are symmetric and
  there exists $\vartheta > 0$ and $\lambda \geq 0$ such that for all $\eta^+ \subset \R^d$ with $|\eta^+| < \infty$
  \begin{align}\label{EQ:07}
   \sum \limits_{x \in \eta^+} \sum \limits_{y \in \eta^+ \backslash x}a^+(x-y) 
   \leq \vartheta \sum \limits_{x \in \eta^+}\sum \limits_{y \in \eta^+ \backslash x}a^-(x-y) + \lambda|\eta^+|.
  \end{align}
\end{enumerate}
Without loss of generality we may normalize $b^- \neq 0$ such that $\| b^- \|_{L^1} = 1$.
Condition \eqref{EQ:07} states that $\vartheta a^- - a^+$ is a stable (pair) potential in the sense of Ruelle.
Below we give some examples of functions $a^{\pm}$.
\begin{Example}
 The following are sufficient for condition (S).
 \begin{enumerate}
  \item Suppose that there exists $\vartheta > 0$ such that $a^+ \leq \vartheta a^-$.
   Such condition was used in \cite{FKK15}.
  \item Let $a^{\pm} \in L^1(\R^d)$ be symmetric, continuous and bounded and assume that $a^+$ has compact support and $a^-(0) > 0$ (see \cite[Proposition 3.7]{KK16}).
  \item Take $a^{\pm}(x) := \frac{c_{\pm}}{(2\pi \sigma_{\pm}^2)^{\frac{d}{2}}} e^{- \frac{|x|^2}{2 \sigma_{\pm}^2}}$ with $c_{\pm}, \sigma_{\pm} > 0$
  (see \cite[Proposition 3.8]{KK16}).
 \end{enumerate}
\end{Example}
\begin{Remark} 
 Note that \eqref{EQ:00} holds in the so-called low-activity regime, i.e. $z$ is small enough.
 Under the given condition (E) it is well-known that there exists a unique Gibbs measure $\mu_{\mathrm{inv}}$ on $\Gamma^-$
 with activity $z$ and potential $\psi$. Moreover the corresponding state evolution for the environment is ergodic with exponential rate (see \cite[Theorem 2]{FK17}).
\end{Remark}

\subsection{Statistical Markov evolutions}
In this work we study Markov dynamics associated to the (scaled) Markov operator $L_{\e} = L^S + \frac{1}{\e}L^E$, $\e > 0$.
Here the action of $L^E$ given by \eqref{EQ:31} is extended onto functions on $\Gamma^2$ by its action only in the variable $\gamma^-$. 
The classical approach to the construction of a Markov process $(\gamma_t^{\e})_{t \geq 0} \subset \Gamma^2$
is based on solving the backward Kolmogorov equation
\begin{align}\label{BKE}
 \frac{\partial F^{\e}_t}{\partial t} = L_{\e}F^{\e}_t, \ \ F_t^{\e}|_{t=0} = F_0, \ \ t \geq 0.
\end{align}
Namely for $F_0 = \1_A$ the solution gives the transition probabilities $p_t(\gamma, A) = F_t(\gamma)$ of the Markov process.
Hence such a process should satisfy the identity
\[
 F_t(\gamma) = \E_{\gamma}(F(\gamma_t^{\e})), \ \ t \geq 0,
\]
where $\gamma$ is the configuration of the process at initial time $t = 0$.

Denote by $\mu_t^{\e}$ the one-dimensional distributions (= state evolution) of such a Markov process. 
Then (heuristically) they satisfy the the Fokker-Planck equation
\begin{align}\label{EQ:05}
 \frac{d }{d t}\int \limits_{\Gamma^2}F(\gamma)d \mu^{\e}_t(\gamma) = \int \limits_{\Gamma^2} L_{\e}F(\gamma) d \mu^{\e}_t(\gamma), \ \ \mu^{\e}_t|_{t=0} = \mu_0, \ \ F \in \mathcal{FP}(\Gamma^2).
\end{align}
In the particular case of constant death rates (i.e. $b^- = a^- = 0$) existence and uniqueness of such a process was studied in \cite{GK06, KS06}.
The general case is still an open challenging mathematical problem of modern probability.
Here and below we restrict ourselves to the study of the Fokker-Planck equation \eqref{EQ:05} without addressing the existence problem 
of an associated Markov process.

It was proposed in \cite{KKM08} to study solutions to \eqref{EQ:05} in the class of states 
for which their associated sequence of correlation functions (see Section 2 for the definition) satisfy a certain (time-homogeneous) Ruelle bound.
The corresponding correlation function evolution should then satisfy a Markov analogue of the well-known BBGKY-hierarchy from physics.
A mathematical realization for general birth-and-death dynamics based on semigroup methods is given in \cite{FKK12, FK16b}.

Applying this results to the particular model described above we deduce that \eqref{EQ:05} is well-posed in the class of states $(\mu_t^{\e})_{t \geq 0}$
for which their associated sequence of correlation functions $k_{\mu_t^{\e}}^{(n,m)}$ satisfy 
for some constants $A_{\mu_0} > 0$ and $\alpha^{+} \in \R$ the (time-homogeneous) Ruelle bound
\begin{align}\label{EQ:51}
 \| k_{\mu_t^{\e}}^{(n,m)}\|_{L^{\infty}((\R^d)^n \times (\R^d)^m)} \leq A_{\mu_0} e^{\alpha^+ n} e^{\alpha_*^- m}, \ \ n,m \geq 0.
\end{align}
Here the evolution of correlation functions $k_{t}^{\e} := (k_{\mu_t^{\e}}^{(n,m)})_{n,m=0}^{\infty}$ 
satisfies a Markov analogue of the BBGKY-hierarchy
\[
 \frac{\partial k_t^{\e}}{\partial t} = L_{\e}^{\Delta}k_t^{\e}, \ \ k_t^{\e}|_{t=0} = k_0.
\]
The operator $L^{\Delta}_{\e}$ is explicitly given in Section 4.
Unfortunately, such semigroup methods require that $m \geq 0$ is sufficiently large (= 'high mortality regime').
Having applications in mind it is feasible to study also the case where $m \geq 0$ is small or even equals zero.

\subsection{Aim of this work}
In this work we study \eqref{EQ:05} under the conditions (S) and (E) where $m \geq 0$ is small or even equals to zero.
As a consequence we cannot expect that \eqref{EQ:51} holds for constants $A_{\mu_0},\alpha^+$ uniform in $t$, i.e. a time-inhomogeneous version of the Ruelle bound
has to be used. Thus we study the following problems:
\begin{enumerate}
 \item[(i)] Prove that, given (S) and (E), the Fokker-Planck equation \eqref{EQ:05} can be uniquely solved 
 in the class of states $(\mu_t^{\e})_{t \geq 0}$ for which the corresponding sequence of correlation functions $(k_{\mu_t^{\e}}^{(n,m)})_{n,m=0}^{\infty}$
 satisfy for some constant $A(t) > 0$ a time-inhomogeneous Ruelle bound
 \[
  k_{\mu_t^{\e}}^{(n,m)} \leq A(t) e^{\alpha^+(t)n} e^{\alpha_*^- m}, \ \ n,m \geq 0,
 \]
 where $A(t)$ is some function of $t$ and $\alpha^+(t)$ is affine linear in $t$.
 Moreover provide reasonable estimates on the functions $A(t)$ and $\alpha^+(t)$.
 \item[(ii)] Study Lyapunov-type estimates for the evolution of states.
\end{enumerate}
Problem (i) was investigated for the BDLP model in \cite{KK16}
where the evolution of states was studied in terms of correlation functions and the latter one was constructed in an increasing scale of Banach spaces.
Weak uniqueness for the Fokker-Planck equation was then obtained in \cite{F16a}.
In this work we provide an extension of this techniques to the case of two-component dynamics.
The obtained correlation function evolution is in our case constructed in an increasing two-parameter scale of Banach spaces.
Concerning problem (ii) we construct a Lyapunov function for the operator $L_{\e}$ and deduce then classically estimates on the state evolution.

The main part of this work is devoted to the study of the limit $\e \to 0$. It falls into the particular class of Random evolution framework
considered e.g. in \cite{ET86},\cite{K92},\cite{PINSKY},\cite{SHS}.
Most of the existing works typically deal with rather simple environments or systems and are studied via the backward Kolmogorov equation.
However, for the model considered in this work both (system and environment) are infinite particle dynamics for which an
analysis of the backward Kolmogorov equation is absent.
New techniques based on the Fokker-Planck equation and corresponding correlation function evolution have to be developed.
In \cite{FK16} we have studied a (general) system of finitely many interacting point particles with rates depending on
an infinite particle equilibrium process (e.g. of Glauber-type). The case of two general birth-and-death evolutions on $\Gamma^2$
was then studied by semigroup methods in \cite{FK17}. 
In order to apply these methods to the particular model given by \eqref{EQ:31} and \eqref{EQ:32} it is necessary 
to work in the 'high mortality regime'. In such a case it can be shown that the population gets extinct with exponential speed, 
i.e. $\mu_t^{\e} \longrightarrow \delta_{\emptyset}$ as $t \to \infty$ (apply e.g. \cite[Theorem 2]{FK17}). 
Hence for most of the interesting cases we cannot directly apply the previous results.

In this work we show how the stochastic averaging principle ($\e \to 0$) can be obtained without assuming the 'high mortality regime'.
More precisely we suppose to study the following problem:
\begin{enumerate}
 \item[(iii)] Suppose that the coupling constant $g$ is small enough, i.e.
 \begin{align}\label{EQ:01}
  0 < g < (m + \lambda)e^{- \alpha_*^-}.
 \end{align}
 Then for any $F \in \mathcal{FP}(\Gamma^+)$ (specified in the next section)
 \[
  \int \limits_{\Gamma^2}F(\gamma^+)d\mu_t^{\e}(\gamma^+,\gamma^-) \longrightarrow \int \limits_{\Gamma^+}F(\gamma^+)d\overline{\mu}_t(\gamma^+), \ \ \e \to 0
 \]
 holds uniformly on compacts in $t$. Moreover the evolution $(\overline{\mu}_t)_{t \geq 0}$ is the unique solution to the Fokker-Planck equation
 \begin{align}\label{EQ:28}
  \frac{d}{dt}\int \limits_{\Gamma^+}F(\gamma^+)d\overline{\mu}_t(\gamma^+) = \int \limits_{\Gamma^+}(\overline{L}F)(\gamma^+)d\overline{\mu}_t(\gamma^+), \ \ \overline{\mu}_t|_{t=0} = \mu_0^+, \ \ F \in \mathcal{FP}(\Gamma^+),
 \end{align} 
 where $\mu_0^+$ is the marginal of $\mu_0$ onto $\Gamma^+$ and 
 \begin{align*}
  (\overline{L}F)(\gamma^+) &= \sum \limits_{x \in \gamma^+}\left( m + g \rho + \sum \limits_{y \in \gamma^- \backslash x}a^-(x-y)\right)(F(\gamma^+ \backslash x) - F(\gamma^+))
  \\ \notag &\ \ \ + \sum \limits_{x \in \gamma^+}\int \limits_{\R^d}a^+(x-y)(F(\gamma^+ \cup y) - F(\gamma^+))d y.
 \end{align*}
 Here $\rho$ is the density (first correlation function) of the Gibbs measure with activity $z > 0$ and interaction potential $\psi$.
 The latter one is constant due to the fact that $\psi$ and $z$ are translation invariant in space.
\end{enumerate}
The corresponding analysis is based on scales of Banach spaces and, in particular, on methods developed in \cite{FK17}, \cite{KK16}, \cite{F16}.
We conclude with a particular example.
\begin{Example}
 Take $\psi = 0$ and $z > 0$, then $L^E$ reduces to the so-called Sourgailis model studied in \cite{F11} with Poisson measure $\mu_{\mathrm{inv}} = \pi_{z}$
 as unique invariant measure. Its density satisfies $\rho = z$ and hence if 
 \[
  g < \frac{\lambda + m}{z},
 \]
 then (E) is satisfied for $\alpha_{*}^- > \log(z)$ and \eqref{EQ:01} holds.
 
 In particular take $m,g$ such that
 \[
  \max\left\{ \frac{\| a \|_{L^1} - \lambda}{2}, 0\right\} < m < \| a \|_{L^1}, \ \ 
  \frac{\| a\|_{L^1} - m}{z} < g < \frac{m+\lambda}{z}.
 \]
 Then (E) and \eqref{EQ:01} are satisfied for any $\alpha_*^- \in (\log(z), \log( \frac{m+\lambda}{g}))$.
 Since $m < \| a \|_{L^1}$ the joint dynamics is overcritical, but since $m + gz > \| a \|_{L^1}$
 the averaged dynamics is subcritical, i.e. it can be studied by semigroup methods and applying \cite[Theorem 2]{FK17} 
 we see that $\overline{\mu}_t \longrightarrow \delta_{\emptyset}$.
 This property is a consequence of the interactions with the environment.
\end{Example}

\subsection{Structure of the work}
This work is organized as follows. 
In order that this work is self-contained we recall some facts of harmonic analysis on configuration spaces in Section 2.
The results of this work are presented and discussed in Section 3.
In Section 4 we study existence and uniqueness to \eqref{EQ:05}.
Estimates on the correlation functions are shown in Section 5 whereas in Section 6 we study Lyapunov-type estimates.
The stochastic averaging principle, i.e. problem (iii), is then proved in Section 7.

\section{Harmonic analysis on the configuration space}

\subsection{One-component case}
Let us first consider only particles of one type. Hence we omit (only here) for simplicity of notation the dependence on the spins $\pm$.
Recall that the one-component configuration space of locally finite configurations is defined by
\begin{align}\label{EQ:02}
 \Gamma = \{ \gamma \subset \R^d \ | \ |\gamma \cap \Lambda| < \infty \ \ \text{ for all compacts } \Lambda \subset \R^d\}.
\end{align}
It is well-known that $\Gamma$ is a Polish spaces with respect to the smallest topology
such that all maps $\gamma \longmapsto \sum_{x \in \gamma}f(x)$ are continuous for any continuous function $f$ with compact support (see \cite{KK06}).
The corresponding Borel-$\sigma$-algebra is then the smallest $\sigma$-algebra such that 
\[
 \Gamma \ni \gamma \longmapsto |\gamma \cap \Lambda| \in \N_0
\]
is measurable for any compact $\Lambda \subset \R^d$. 
We call probability measures on $\Gamma$ states and measurable functions $F$ on $\Gamma$ observables.
\begin{Example}
 The Poisson measure $\pi_{z}$ with intensity measure $zdx$, $z > 0$, is the most prominent example of a purely chaotic state.
 It is uniquely determined by 
 \[
  \pi_z( \{ \gamma \in \Gamma \ | \ |\gamma \cap \Lambda| = n \}) = \frac{m(\Lambda)^n z^n}{n!}e^{-z m(\Lambda)},
 \]
 where $n \geq 0$ and $\Lambda \subset \R^d$ is a compact. Here $m(\Lambda)$ denotes the Lebesgue measure of $\Lambda$.
\end{Example}
Below we describe a particular subclass of states described in terms of their associated correlation functions.
These ideas go back to the works of Lennard \cite{L74,L75}.

The space of finite configurations 
\[
 \Gamma_0 := \{ \eta \subset \R^d \ | \ |\eta| < \infty\}
\]
is equipped with the smallest $\sigma$-algebra such that $\Gamma_0 \ni \eta \longmapsto |\eta \cap \Lambda|$
is measurable for any compact $\Lambda \subset \R^d$.
\begin{Remark}
 Note that any measurable function $G: \Gamma_0 \longrightarrow \R$ is uniquely determined by its associated sequence of symmetric measurable
 coordinate functions $G^{(n)}: (\R^d)^n \longrightarrow \R$ via
 \begin{align}\label{EQ:19}
  G^{(n)}(x_1,\dots, x_n) = \begin{cases}G(\{x_1,\dots, x_n\}), & x_i \neq x_j, i \neq j \\ 0, & \text{ otherwise}\end{cases}.
 \end{align}
 Note that $G^{(0)} = G(\emptyset)$ is simply a constant.
\end{Remark}
We define a  measure $\lambda$ on $\Gamma_0$ by the relation
\begin{align}\label{EQ:50}
 \int \limits_{\Gamma_0}G(\eta)d\lambda(\eta) = G(\emptyset) + \sum \limits_{n=1}^{\infty}\frac{1}{n!}\int\limits_{\R^{dn}}G(\{x_1,\dots, x_n\})dx_1\dots dx_n,
\end{align}
where $G$ is any non-negative measurable function. 
The following identity is an combinatorial analogue of the integration by parts formula (see \cite{F17WR} for a proof).
\begin{Lemma}
 For any measurable function $G: \Gamma_0 \times \Gamma_0 \times \Gamma_0 \longrightarrow \R$ it holds that
 \begin{align}\label{IBP}
  \int \limits_{\Gamma_0}\sum \limits_{\xi \subset \eta}G(\xi, \eta \backslash \xi, \eta)d \lambda(\eta)
  = \int \limits_{\Gamma_0}\int \limits_{\Gamma_0}G(\xi, \eta, \eta \cup \xi)d \lambda(\xi)d \lambda(\eta)
 \end{align}
 provided on side of the equality is finite for $|G|$.
\end{Lemma}
Let $B_{bs}(\Gamma_0)$ be the space of all bounded measurable functions with bounded support, i.e.
$G \in B_{bs}(\Gamma_0)$ iff $G$ is bounded and there exists a compact $\Lambda \subset \R^d$ and $N \in \N$ with $G(\eta) = 0$, 
whenever $\eta \cap \Lambda^c \neq \emptyset$ or $|\eta| > N$. 
Equivalently, its associated sequence of symmetric functions $G^{(n)}$ (see \eqref{EQ:19}) 
satisfy $G^{(n)} = 0$ for all $n > N$ and $G^{(1)}, \dots, G^{(N)}$ are bounded with compact support.

The $K$-transform is, for $G \in B_{bs}(\Gamma_0)$, defined by
\begin{align}\label{KTRANSFORM}
 (KG)(\gamma) := \sum \limits_{\eta \Subset \gamma}G(\eta), \ \ \gamma \in \Gamma,
\end{align}
where $\Subset$ means that the sum is taken over all finite subsets $\eta$ of $\gamma$.
Note that due to $G \in B_{bs}(\Gamma_0)$ only finitely many terms in the sum are non-vanishing.

Let $\mu$ be a probability measure on $\Gamma$ with finite local moments, i.e. $\int_{\Gamma}|\gamma \cap \Lambda|^nd\mu(\gamma) < \infty$
for all compacts $\Lambda$ and $n \geq 1$. 
\begin{Definition}
 The correlation function $k_{\mu}: \Gamma_0 \longrightarrow \R_+$ is (uniquely) defined by the relation
\begin{align*}
 \int \limits_{\Gamma}(KG)(\gamma)d\mu(\gamma) = \int \limits_{\Gamma_0}G(\eta)k_{\mu}(\eta)d\lambda(\eta), \ \ G \in B_{bs}(\Gamma_0).
\end{align*}
\end{Definition}
The assumption on finite local moments can be used to show that the integrand on the right-hand side is, indeed, integrable in $\lambda$.
Note that not every state $\mu$ admits a correlation function $k_{\mu}$. It is necessary and sufficient that $\mu$
is locally absolutely continuous with respect to the Poisson measure (see \cite{KK02} for additional details).
\begin{Remark}
 The Poisson measure $\pi_z$ has correlation function $k_{\pi_z}(\eta) = z^{|\eta|}$.
\end{Remark}

\subsection{Two-component case}
Below we briefly describe an extension to the two-component case.
Let $\Gamma^{\pm}$ be two independent copies of $\Gamma$ as given in \eqref{EQ:02}.
The two-component configuration space is defined as the product $\Gamma^2 = \Gamma^+ \times \Gamma^-$ and it is equipped with the product topology.
All results explained above naturally extend to this case. 
In order to fix the notation, a brief summary is given below.

Let $\Gamma_0^2 := \Gamma_0^+ \times \Gamma_0^-$, $\eta := (\eta^+, \eta^-)$ and $|\eta| := |\eta^+| + |\eta^-|$.
Define $G \in B_{bs}(\Gamma_0^2)$ iff $G$ is bounded, measurable
and there exists a compact $\Lambda \subset \R^d$ and $N \in \N$ such that $G(\eta) = 0$, whenever $|\eta| > N$ or $\eta^{\pm} \cap \Lambda^c \neq \emptyset$.

The $K$-transform is, for $G \in B_{bs}(\Gamma_0^2)$, defined by
\[
 (\mathbbm{K}G)(\gamma) := \sum_{\eta \Subset \gamma}G(\eta),
\]
where $\xi \subset \eta$ and $\eta \Subset \gamma$ are defined component-wise. 

Let $\mu$ be a probability measure on $\Gamma^2$ with finite local moments, i.e.
\[
 \int_{\Gamma^2}|\gamma^- \cap \Lambda|^n |\gamma^+ \cap \Lambda|^n d\mu(\gamma) < \infty
\]
for all compacts $\Lambda$ and $n \geq 1$.
Analogously to the one-component case we define the correlation function $k_{\mu}$, provided it exists, by
\begin{align}\label{EQ:30}
 \int \limits_{\Gamma^2}\mathbbm{K}G(\gamma)d \mu(\gamma) = \int \limits_{\Gamma_0^2}G(\eta)k_{\mu}(\eta)d \lambda(\eta), \ \ G \in B_{bs}(\Gamma_0^2).
\end{align}
Here $\lambda = \lambda^+ \otimes \lambda^-$, where $\lambda^{\pm}$ are defined by \eqref{EQ:50}.
By abuse of notation we let $\lambda$ stand for the corresponding measure on $\Gamma_0^{\pm}$ or $\Gamma_0^2$, respectively.
At this point it is worth to mention that not every non-negative function $k$ on $\Gamma_0^2$ is the correlation function of some state $\mu$.
It is necessary and sufficient that $k(\emptyset) = 1$ and that $k$ is positive definite in the sense of Lenard, i.e.
\[
 \int \limits_{\Gamma_0^2}G(\eta)k(\eta)d\lambda(\eta) \geq 0, \ \ G \in B_{bs}(\Gamma_0^2), \ \ \mathbbm{K}G \geq 0,
\]
see \cite{L74,L75,KK02}. Some additional details for states on $\Gamma^2$ are discussed in \cite{F11a}.
\begin{Example}
 The two-component Poisson measure on $\Gamma^2$ with activity parameters $z^{\pm} > 0$ is defined as the product measure 
 $\pi_{z^+,z^-} := \pi_{z^+} \otimes \pi_{z^-}$. Hence it has correlation function $k_{\pi_{z^+,z^-}}(\eta) = (z^+)^{|\eta^+|}(z^-)^{|\eta^-|}$.
\end{Example}

\section{Statement of the results}

\subsection{Joint dynamics}
Let $\e > 0$ be fixed. Observe that $C_{\psi}$ is continuous in $\alpha$. 
Hence by \eqref{EQ:00} we find $\alpha^{*,-} > \alpha_*^-$ such that for all $\alpha^- \in [\alpha_*^-,\alpha^{*,-})$ we have
\begin{align}\label{ENV:COND}
  z e^{-\alpha^-}C_{\psi}(\alpha^-) < 1.
\end{align}
Moreover, set
\begin{align}\label{EQ:13}
 \alpha_*^+ := \begin{cases}\log\left( \max\left\{ \frac{\lambda}{\lambda+m}, \vartheta \right\} \right), & \lambda > 0 \\ \log(\vartheta), & \lambda = 0 \end{cases}
\end{align}
and let $\mathcal{FP}(\Gamma^2) = K(B_{bs}(\Gamma_0^2))$. For $\alpha = (\alpha^+, \alpha^-)$ let $\K_{\alpha}$
be the Banach space of all equivalence classes of functions $k: \Gamma_0^2 \longrightarrow \R$ with finite norm 
\[
 \Vert k \Vert_{\K_{\alpha}} := \esssup \limits_{\eta \in \Gamma_0^2}\ |k(\eta)|e^{- \alpha^+|\eta^+|}e^{- \alpha^-|\eta^-|}.
\]
Denote by $\mathcal{P}_{\alpha}$ the space of all states $\mu$ such that $k_{\mu} \in \K_{\alpha}$. 
Moreover, let $\mathcal{P} = \bigcup_{\alpha \in \R^2}\mathcal{P}_{\alpha}$.
For simplicity of notation we let $\langle F, \mu \rangle = \int_{\Gamma^2}F(\gamma)d\mu(\gamma)$ where $\mu$ is a state and $F$ an integrable function on $\Gamma^2$.
Existence and uniqueness for solutions to \eqref{EQ:05} is stated below.
\begin{Theorem}\label{TH:03}
 Let $\alpha \in \R^2$ be such that $\alpha^- \in [\alpha_*^-, \alpha^{*,-})$ and $\alpha^+ > \alpha_*^+$.
 Then for any $\mu_0 \in \mathcal{P}_{\alpha}$ there exists $(\mu_t^{\e})_{t \geq 0} \subset \mathcal{P}$ such that for any $F \in \mathcal{FP}(\Gamma^2)$
 \begin{enumerate}
  \item[(i)] $L_{\e}F \in L^1(\Gamma^2,d\mu_t^{\e})$ and $t \longmapsto \langle F, \mu_t^{\e} \rangle$ is continuous.
  \item[(ii)] $t \longmapsto \langle F, \mu_t^{\e}\rangle$ is continuously differentiable and \eqref{EQ:05} holds for all $t \geq 0$.
 \end{enumerate}
 Moreover, given $(\nu_t^{\e})_{t \geq 0} \subset \mathcal{P}$ such that for all $F \in \mathcal{FP}(\Gamma^2)$ the following properties hold
 \begin{enumerate}
  \item[(i)] $L_{\e}F \in L^1(\Gamma^2,d\nu_t^{\e})$ and $t \longmapsto \langle F, \nu_t^{\e} \rangle$ is locally integrable.
  \item[(ii)] $t \longmapsto \langle F, \mu_t\rangle$ is absolutely continuous and \eqref{EQ:05} holds for a.a. $t \geq 0$.
  \item[(iii)] For each $T > 0$ there exists $\beta^+ > \alpha^+$ such that
  \[
   \sup \limits_{t \in [0,T]} \ \Vert k_{\nu_t^{\e}} \Vert_{\K_{\beta^+,\alpha^-}} < \infty.
  \]
 \end{enumerate}
 Then $\mu_t^{\e} = \nu_t^{\e}$ for all $t \geq 0$.
\end{Theorem}
The construction of solutions to the Fokker-Planck equation is based on the construction of an correlation function evolution.
Inspecting the proof we easily deduce the following.
\begin{Remark}
 The proof shows that \eqref{EQ:05} is equivalent to a weak formulation of
 \begin{align}\label{CORR:00}
  \frac{\partial k_t}{\partial t} = L_{\e}^{\Delta}k_{t}, \ \ k_{t}|_{t= 0} = k_{0},
 \end{align}
 where $L_{\e}^{\Delta}$ is explicitly given in Section 3. The latter one is an Markov analogue of the well-known BBGKY-hierarchy from physics.
 Moreover for any $T > 0$ the evolution $[0,T) \ni t \longmapsto k_{\mu_t^{\e}} \in \K_{\beta(\alpha,T)}$
 is continuously differentiable and the unique classical solution to \eqref{CORR:00} in $\K_{\beta(\alpha,T)}$,
 where $\beta(\alpha,T) = (\beta^+(\alpha,T), \alpha^-)$ with $\beta^+(T,\alpha) = \alpha^+ + (\lambda + \| a^+ \|_{L^1})T$.
\end{Remark}
Here and below we let $E_{b^-}(\eta^+,\eta^-) := \sum_{x \in \eta^+}\sum_{y \in \eta^-}b^-(x-y)$.
The next statement is an extension of \cite{KK16} and gives bounds on the correlation function evolution.
\begin{Theorem}\label{APRIORI}
 Let $\alpha \in \R^2$ be such that 
 $\alpha^- \in [\alpha_*^-, \alpha^{*,-})$, $\alpha^+ > \alpha_*^+$ and take $\mu_0 \in \mathcal{P}_{\alpha}$. 
 Let $(\mu_t^{\e})_{t \geq 0} \subset \mathcal{P}$ be the corresponding evolution of states and $(k_{\mu_t^{\e}})_{t \geq 0}$ its evolution of correlation functions.
 Then the following assertions hold.
 \begin{enumerate}
  \item[(a)] Case $m \leq \| a^+ \|_{L^1}$. 
  \\ If $\lambda = 0$, take any $\delta > 0$ and $\alpha_{\delta}^+ = \alpha^+$.
  If $\lambda > 0$, take any $\delta \in (0,\lambda + m)$ and $\alpha_{\delta}^+ \geq \alpha^+$ with $\alpha_{\delta}^+ > \log\left( \frac{\lambda}{\delta} \right)$. 
  Then for any $\eta \in \Gamma_0^2$ and $t \geq 0$
  \[
   k_{\mu_t^{\e}}(\eta) \leq \Vert k_{\mu_0} \Vert_{\K_{\alpha}} e^{\alpha_{\delta}^+|\eta^+|}e^{\alpha^-|\eta^-|} e^{\left( \| a^+ \|_{L^1} + \delta - m\right)|\eta^+|t}e^{-g t E_{b^-}(\eta^+,\eta^-)}.
  \]
  \item[(b)] Case $m > \| a^+ \|_{L^1}$.
  \\ Let $\delta \in (0, m - \| a^+ \|_{L^1})$. If $\lambda = 0$, take $\alpha_{\delta}^+ = \alpha^+$.
  If $\lambda > 0$, take $\alpha_{\delta}^+ \geq \alpha^+$ such that $\alpha_{\delta}^+ \geq \log\left( \frac{\lambda}{m - \| a^+ \|_{L^1} - \delta} \right)$.
  Then for any $\eta \in \Gamma_0^2$ and $t \geq 0$
  \[
   k_{\mu_t^{\e}}(\eta) \leq \Vert k_{\mu_0} \Vert_{\K_{\alpha}} e^{\alpha_{\delta}^+|\eta^+|}e^{\alpha^-|\eta^-|} e^{-g t E_{b^-}(\eta^+,\eta^-)}e^{-\delta t}.
  \]
 \end{enumerate}
\end{Theorem}
Assertion (a) shows that the evolution of correlation functions constructed in a scale of Banach spaces is, in general, not localized in one Banach space
but belongs to larger ones when $t$ increases.
The additional factor $e^{-tg E_{b^-}(\eta^+,\eta^-)}$ reflects the interactions with the environment.
\begin{Remark}\label{REMARK:00}
 It is worth to mention that both statements include also the case $\e = \infty$ with $\frac{1}{\infty} := 0$.
 Such a choice for $\e$ describes the BDLP dynamics interacting with a stationary environment. 
 In particular, by taking $g = 0$ we recover the state evolution for the isolated BDLP dynamics obtained in \cite{KK16}.
\end{Remark}
Below we give estimates on the state evolution in terms of Lyapunov functions.
Let us assume the following condition.
\begin{enumerate}
 \item[(L)] There exists an integrable, symmetric function $e: \R^d \longrightarrow (0,1]$ such that $e(x+y) \leq \frac{e(x)}{e(y)}$ holds and
 \begin{align}\label{EQ:29}
   \left \Vert \frac{a^+}{e} \right \Vert_{L^1} < \infty, \ \ \ \ \left \Vert \frac{a^+}{e} \right \Vert_{L^{\infty}} < \infty.
 \end{align}
 Moreover, suppose that $\psi$ is integrable.
\end{enumerate}
The following is our main guiding example.
\begin{Example}
 Take $e(x) := e^{- \delta|x|}$ for some $\delta > 0$.
 Then \eqref{EQ:29} holds, provided $a^+(x) \leq c e^{\delta|x|}$ for some constant $c > 0$ and
 $\int_{\R^d}a(x)e^{\delta |x|} dx < \infty$.
\end{Example}
Let $\kappa \in (0,d)$ and set
\[
 \Xi(x,y) := e(x)e(y) \frac{1 + |x-y|^{\kappa}}{|x-y|^{\kappa}}, \ \ x \neq y.
\]
The particular choice of $\kappa$ implies $\Xi \in L^1(\R^d \times \R^d)$. Let $\mathbb{V}: \Gamma^2 \longrightarrow [0,\infty]$ be given by
\[
 \mathbb{V} = V_0^+ + V_0^- + V_1^+ + V_1^- + W
\]
with $V_0^{\pm}(\gamma) = \sum_{x \in \gamma^{\pm}}e(x)$, $V_1^{\pm}(\gamma) = \frac{1}{2}\sum_{x \in \gamma^{\pm}}\sum_{y \in \gamma^{\pm}\backslash x}\Xi(x,y)$
and $W(\gamma^+, \gamma^-) = \sum_{x \in \gamma^+}\sum_{y \in \gamma^-}\Xi(x,y)$.
Then we prove the following.
\begin{Theorem}\label{LYP:TH:01}
 Suppose that (L) holds. Then there exists $\Gamma_{\infty}^2 \subset \Gamma^2$ described in terms of $\mathbb{V}$ and the interaction rates $a^{\pm},b^-$
 such that $\mu(\Gamma_{\infty}^2) = 1$ holds for any $\mu \in \mathcal{P}$ and $(L_{\e}\mathbb{V})(\gamma)$ is well-defined for any $\gamma \in \Gamma_{\infty}^2$.
 Moreover the following properties hold:
 \begin{enumerate}
  \item[(a)] There exists a constant $c_{\e} > 0$ such that
  \begin{align*}
   (L_{\e}\mathbb{V})(\gamma) &\leq c_{\e} \mathbb{V}(\gamma) + \frac{z}{\e}\Vert e \Vert_{L^1}, \ \ \gamma \in \Gamma_{\infty}^2.
  \end{align*}
  \item[(b)] Let $(\mu_t^{\e})_{t \geq 0}$ be the evolution of states obtained from \eqref{EQ:05}. Then
  \begin{align*}
   \langle \mathbb{V}, \mu_t^{\e} \rangle \leq \left( \langle \mathbb{V}, \mu_0 \rangle + t\frac{z}{\e}\Vert e \Vert_{L^1} \right)\exp\left(c_{\e} t \right), \ \ t \geq 0.
  \end{align*}
\end{enumerate}
\end{Theorem}
Note that these estimates also give corresponding estimates on the correlation functions.
\begin{Remark}
 Proceeding in the same way as in the proof of Theorem \ref{LYP:TH:01} we can also prove estimates with constants independent of $\e > 0$.
 Namely one has
 \[
   \langle V_0^+, \mu_t^{\e} \rangle \leq \langle V_0^+, \mu_0 \rangle \exp\left( \left(\left \Vert \frac{a^+}{e} \right \Vert_{L^1} - m \right)t \right)
 \]
 and $\langle V^+, \mu_t^{\e} \rangle \leq \langle V^+, \mu_0 \rangle \exp\left( a t \right)$, where $V^+ = V_0^+ + V_1^+$ and
 \begin{align*}
   a &= \max\left\{\sup \limits_{w \in \R^d} \int \limits_{\R^d}\frac{a^+(y)}{e(y)} \frac{1 + |y - w|^{\kappa}}{|y - w|^{\kappa}}dy , \left \Vert \frac{a^+}{e} \right \Vert_{L^1} + \int_{\R^d}a^+(y) \frac{1 + |y|^{\kappa}}{|y|^{\kappa}}dy - m \right\}.
 \end{align*}
\end{Remark}

\subsection{Averaged dynamics}
Let $\mathcal{FP}(\Gamma^+) = K(B_{bs}(\Gamma_0^+))$. 
Similarly to $\K_{\alpha}$ we let $\K_{\alpha^+}$ be the Banach space of equivalence classes of functions with finite norm
\[
 \| k \|_{\K_{\alpha^+}} = \esssup \limits_{\eta^+ \in \Gamma_0^+} |k(\eta^+)|e^{- \alpha^+|\eta^+|}.
\]
Then let $\mathcal{P}_{\alpha^+}$ be the collection of all states $\mu$ on $\Gamma^+$ such that $k_{\mu} \in \K_{\alpha^+}$.
Finally let $\mathcal{P}_+ = \bigcup_{\alpha^+ \in \R} \mathcal{P}_{\alpha^+}$.
Below we give existence and uniqueness for the Fokker-Planck equation associated to $\overline{L}$.
It can be deduced from Theorem \ref{TH:03} by taking $\e = \infty$ and $g = 0$ (see Remark \ref{REMARK:00}).
\begin{Theorem}\label{TH:04}
 Let $\alpha^+ > \alpha_*^+$ and $\mu_0^+ \in \mathcal{P}_{\alpha^+}$.
 Then there exists $(\overline{\mu}_t)_{t \geq 0} \subset \mathcal{P}_+$ such that for each $F \in \mathcal{FP}(\Gamma^+)$ the following properties hold
 \begin{enumerate}
  \item[(i)] $\overline{L}F \in L^1(\Gamma^+, d\overline{\mu}_t)$, $t \longmapsto \int_{\Gamma^+}(\overline{L}F)(\gamma^+)d\overline{\mu}_t(\gamma^+)$ is continuous.
  \item[(ii)] $t \longmapsto \int_{\Gamma^+}F(\gamma^+)d\overline{\mu}_t(\gamma^+)$ is continuously differentiable such that \eqref{EQ:28} holds.
 \end{enumerate}
 Given $(\overline{\nu}_t)_{t \geq 0} \subset \mathcal{P}_+$ such that for each $F \in \mathcal{FP}(\Gamma^+)$ 
 \begin{enumerate}
  \item[(i)] $\overline{L}F \in L^1(\Gamma^+, d\overline{\nu}_t)$, $t \longmapsto \int_{\Gamma^+}(\overline{L}F)(\gamma^+)d\overline{\nu}_t(\gamma^+)$ is continuous.
  \item[(ii)] $t \longmapsto \int_{\Gamma^+}F(\gamma^+)d\overline{\nu}_t(\gamma^+)$ is absolutely continuous and \eqref{EQ:28} holds for a.a. $t \geq 0$.
  \item[(iii)] For all $T > 0$ there exists $\beta^+ > \alpha^+$ such that 
  \[
   \sup \limits_{t \in [0,T]}\| k_{\overline{\nu}_t} \|_{\K_{\beta^+}} < \infty.
  \]
 \end{enumerate}
 Then $\overline{\mu}_t = \overline{\nu}_t$ for all $t \geq 0$.
\end{Theorem}
A priori estimates on the evolution of correlation functions is given below.
\begin{Theorem}\label{APRIORI1}
 Let $\alpha^+ > \alpha_*^+$, $\mu_0^+ \in \mathcal{P}_{\alpha^+}$ and denote by 
 $(\overline{\mu}_t)_{t \geq 0} \subset \mathcal{P}_+$ be the corresponding evolution of states with correlation functions $k_{\overline{\mu}_t}$.
 Then the following assertions hold.
 \begin{enumerate}
  \item[(a)] Case $m + g\rho \leq \| a^+ \|_{L^1}$.
  \\ If $\lambda = 0$, take $\delta > 0$ and $\alpha_{\delta}^+ = \alpha^+$.
  If $\lambda > 0$, take $\delta \in (0,\lambda + m + g \rho)$ and $\alpha_{\delta}^+ \geq \alpha^+$ with $\alpha_{\delta}^+ > \log\left( \frac{\lambda}{\delta} \right)$. 
  Then
  \[
   k_{\overline{\mu}_t}(\eta^+) \leq \Vert k_{\mu_0^+} \Vert_{\K_{\alpha^+}}e^{\alpha_{\delta}^+|\eta^+|} e^{\left( \| a^+ \|_{L^1} + \delta - m - g \rho\right)|\eta^+|t}, \ \ \eta^+ \in \Gamma_0, \ t \geq 0.
  \]
  \item[(b)] Case $\| a^+ \|_{L^1} < m + g \rho$.
  \\ Let $\delta \in (0, m + g \rho - \| a^+ \|_{L^1})$. If $\lambda = 0$, take $\alpha_{\delta}^+ = \alpha^+$.
  If $\lambda > 0$, take $\alpha_{\delta}^+ \geq \alpha^+$ such that $\alpha_{\delta}^+ \geq \log\left( \frac{\lambda}{m + g \rho - \| a^+ \|_{L^1} - \delta} \right)$. Then
  \[
   k_{\overline{\mu}_t}(\eta) \leq \Vert k_{\mu_0^+} \Vert_{\K_{\alpha^+}}e^{\alpha_{\delta}^+|\eta^+|} e^{- \delta t}, \ \ \eta^+ \in \Gamma_0, \ t \geq 0.
  \]
 \end{enumerate}
\end{Theorem}
Note that the critical value has changed from $m = \| a^+ \|_{L^1}$ (see Theorem \ref{APRIORI}) to $m + g \rho = \| a^+ \|_{L^1}$.
This is, of course, a consequence of the additional competition with particles from the environment.

Below we give similar Lyapunov estimates as for the joint dynamics. They can be deduced from Theorem \ref{LYP:TH:01} by taking $\e = \infty$ and $g = 0$.
Let $V: \Gamma^+ \longrightarrow [0,\infty]$ be given by
\[
 V(\gamma^+) = \sum \limits_{x \in \gamma^+}e(x) + \frac{1}{2}\sum \limits_{x \in \gamma^+}\sum \limits_{y \in \gamma^+ \backslash x}\Xi(x,y)
\]
where $e, \Xi$ are as before. Set 
\[
 \Gamma_{\infty}^+ := \left\{ \gamma^+ \in \Gamma^+ \ | \ V(\gamma^+) < \infty, \ \sum \limits_{x \in \gamma^+}\left( \sum \limits_{w \in \gamma^+ \backslash x}a^-(x-w)\right)\left(\sum \limits_{y \in \gamma^+ \backslash x}\Xi(x,y)\right) < \infty \right\}.
\]
As before we may show that $\mu(\Gamma_{\infty}^+) = 1$ holds for any $\mu \in \mathcal{P}_+$ and, moreover,
$(\overline{L}V)(\gamma^+)$ is well-defined for any $\gamma^+ \in \Gamma_{\infty}^+$.
\begin{Theorem}
 Suppose that $\frac{a^+}{e} \in L^1 \cap L^{\infty}$. Then we can find a constant $c > 0$ such that
 \[
  (\overline{L}V)(\gamma^+) \leq c V(\gamma^+), \ \ \gamma^+ \in \Gamma_{\infty}^+.
 \]
 Let $\mu_0^+ \in \mathcal{P}_{\alpha^+}$ with $\alpha^+ > \alpha_*^+$ 
 and let $(\overline{\mu}_t)_{t \geq 0} \subset \mathcal{P}_+$ be the solution to \eqref{EQ:28}. Then 
 \[
  \langle V, \overline{\mu}_t \rangle \leq \langle V, \mu_0^+ \rangle e^{c t}, \ \ t \geq 0.
 \]
\end{Theorem}

\subsection{Stochastic averaging principle}
For a given state $\mu_0 \in \mathcal{P}$ let $\mu_0^+ \in \mathcal{P}_+$ be the marginal on $\Gamma^+$ defined by
\begin{align*}
 \int \limits_{\Gamma^+}F(\gamma^+)d\mu_0^{+}(\gamma^+) = \int \limits_{\Gamma^2}F(\gamma)d\mu_0(\gamma), \ \ F \in \mathcal{FP}(\Gamma^+).
\end{align*}
The following is our main result on problem (iii).
\begin{Theorem}\label{TH:05}
 Let $\alpha^{*,-}$, $\alpha_*^-$ and $\widetilde{\alpha}_*^+$ be such that \eqref{ENV:COND} and
 \begin{align}\label{EQ:47}
  e^{\alpha^{*,-}}g + e^{- \widetilde{\alpha}_*^+}\lambda \leq \lambda + m
 \end{align}
 hold. Let $\alpha^- \in [\alpha_*^-, \alpha^{*,-})$, $\alpha^+ > \widetilde{\alpha}_*^+$ and $\mu_0 \in \mathcal{P}_{\alpha}$.
 Denote by $(\mu_t^{\e})_{t \geq 0} \subset \mathcal{P}$ the evolution of states given by Theorem \ref{TH:03}
 and by $(\overline{\mu}_t)_{t \geq 0} \subset \mathcal{P}_+$ be the evolution of states given by Theorem \ref{TH:04}
 with initial state $\mu_0^+ \in \mathcal{P}_{\alpha^+}$. Then 
 \[
  \int \limits_{\Gamma^2}F(\gamma^+)d \mu_t^{\e}(\gamma^+,\gamma^-) \longrightarrow \int \limits_{\Gamma^+}F(\gamma^+)d \overline{\mu}_t(\gamma^+), \ \ F \in \mathcal{FP}(\Gamma^+).
 \]
 holds uniformly on compacts w.r.t. $t \geq 0$ as $\e \to 0$.
\end{Theorem}
\begin{Remark}
 Suppose that $0 < g < (\lambda + m)e^{- \alpha_*^-}$ holds. Then we can find $\alpha^{*,-}, \widetilde{\alpha}_*^+ \in \R$ 
 such that \eqref{ENV:COND} and \eqref{EQ:47} are satisfied.
\end{Remark}

\section{Proof: Theorem \ref{TH:03}}
Here and below let $\e > 0$ be fixed. Note that $\e := \infty$ is allowed, provided we set $\frac{1}{\infty} := 0$.
The proof extends some ideas and techniques developed in \cite{F17WR, KK16, FK17}.
For $\delta > 0$ let $L_{\delta,\e} = L_{\delta}^S + \frac{1}{\e}L_{\delta}^E$ be given by
\begin{align*}
 (L_{\delta}^SF)(\gamma) &= \sum \limits_{x \in \gamma^+}\left( m + \sum \limits_{y \in \gamma^+ \backslash x}a^-(x-y) + g\sum \limits_{y \in \gamma^-}b^-(x-y)\right)(F(\gamma^+ \backslash x, \gamma^-) - F(\gamma))
 \\ &\ \ \ + \sum \limits_{x \in \gamma^+}R_{\delta}(x)\int \limits_{\R^d}a^+(x-y)(F(\gamma^+ \cup y, \gamma^-) - F(\gamma^+, \gamma^-))d y
 \\ (L_{\delta}^EF)(\gamma^-) &= \sum \limits_{x \in \gamma^-}(F(\gamma^- \backslash x) - F(\gamma^-)) + \int \limits_{\R^d}z_{\delta}(x)e^{-E_{\psi}(x,\gamma^-)}(F(\gamma^- \cup x) - F(\gamma^-))d x,
\end{align*}
where $R_{\delta}(x) = e^{- \delta|x|^2}$ and $z_{\delta}(x) := z R_{\delta}(x)$.
Note that $\delta = 0$ corresponds to the original model whereas $\delta > 0$ describes a model with integrable birth rate.
It has the following property. 

Let $\delta > 0$ and $\mu_0 \in \mathcal{P}$ be such that $\mu_0(\Gamma_0^2) = 1$. 
Then one can show that $d\mu_0(\gamma) = \1_{\Gamma_0^2}(\gamma)h_0(\gamma)d\lambda(\gamma)$ holds for some probability density $h_0$ on $\Gamma_0^2$.
Moreover, the corresponding evolution of states exists (it is constructed below) and satisfies 
$d\mu_t^{\delta,\e}(\gamma) = \1_{\Gamma_0^2}(\gamma)h_t^{\delta,\e}(\gamma)d\lambda(\gamma)$ with $\int_{\Gamma_0^2}h_t^{\delta,\e}(\eta)d\lambda(\eta) = 1$.
Roughly speaking Theorem \ref{TH:03} is deduced by taking the limit $\delta \to 0$. 
The rigorous limit transition is performed on the level of correlation functions.

\subsection{Evolution of integrable densities}
We construct an evolution of densities on $L^1(\Gamma_0^2, d\lambda)$ corresponding to the Markov operator $L_{\delta,\e}$ with $\delta > 0$.
For $\eta \in \Gamma_0^2$ let
\[
 D_{\delta,\e}(\eta) := m|\eta^+| + \sum \limits_{x \in \eta^+}\sum \limits_{y \in \eta^+ \backslash x}a^-(x-y) 
 + g\sum \limits_{x \in \eta^+}\sum \limits_{y \in \eta^-}b^-(x-y)
 + \frac{1}{\e}|\eta^-| + \frac{1}{\e}\| z_{\delta}e^{- E_{\psi}(\cdot, \eta^-)}\|_{L^1},
\]
where $\| z_{\delta}e^{- E_{\psi}(\cdot, \eta^-)}\|_{L^1} = \int_{\R^d}z_{\delta}(x)e^{-E_{\psi}(x,\eta^-)}d x$. 
We consider this function as a multiplication operator on $L^1(\Gamma_0^2,d \lambda)$ with domain
$\mathcal{D}_{\delta,\e} = \left \{ h \in L^1(\Gamma_0^2, d\lambda) \ | \ D_{\delta,\e} \cdot h \in L^1(\Gamma_0^2, d\lambda) \right \}$.
Then
\begin{align*}
 (\mathcal{Q}_{\delta,\e}h)(\eta) &= \int \limits_{\R^d}\left( m + \sum \limits_{y \in \eta^+}a^-(x-y) + g \sum \limits_{y \in \eta^-}b^-(x-y)\right)h(\eta^+ \cup x, \eta^-)d x
 \\ &\ \ \ + \sum \limits_{x \in \eta^+} \sum \limits_{y \in \eta^+ \backslash x}R_{\delta}(y)a^+(x-y)h(\eta^+ \backslash y, \eta^-)
 \\ &\ \ \ + \frac{1}{\e}\int \limits_{\R^d}h(\eta^+, \eta^- \cup x)d x + \frac{1}{\e}\sum \limits_{x \in \eta^-}z_{\delta}(x)e^{-E_{\psi}(x,\eta^- \backslash x)}h(\eta^+, \eta^- \backslash x)
\end{align*}
defines a linear operator on $\mathcal{D}_{\delta,\e}$.
Using \eqref{IBP} one can show that for each bounded measurable function $F: \Gamma_0^2 \longrightarrow \R$ and each $h \in \mathcal{D}_{\delta,\e}$ 
\begin{align}\label{EQ:06}
 \int \limits_{\Gamma_0^2}(L_{\delta,\e}F)(\eta)h(\eta)d \lambda(\eta) = \int \limits_{\Gamma_0^2}F(\eta)\left( - D_{\delta,\eta}(\eta)h(\eta) + (\mathcal{Q}_{\delta,\e}h)(\eta)\right)d \lambda(\eta).
\end{align}
\begin{Lemma}\label{PROP:01}
 $(- \mathcal{D}_{\delta,\e} + \mathcal{Q}_{\delta,\e}, \mathcal{D}_{\delta,\e})$ is closable and 
 its closure $(\mathcal{J}_{\delta,\e}, D(\mathcal{J}_{\delta,\e}))$ is the generator of a stochastic semigroup on $L^1(\Gamma_0^2, d\lambda)$. 
 In particular, for each $0 \leq h_0 \in D(\mathcal{J}_{\delta,\e})$ there exists a unique classical solution 
 $0 \leq h_t^{\delta,\e} \subset D(\mathcal{J}_{\delta,\e})$ to
 \[
  \frac{\partial h_t^{\delta,\e}}{\partial t} = \mathcal{J}_{\delta,\e}h_t^{\delta,\e}, \ \ h_{t}^{\delta,\e}|_{t=0} = h_0, \ \ t \geq 0.
 \]
\end{Lemma}
\begin{proof}
 The same arguments as \cite[Lemma 5]{F17WR}, i.e. apply \cite[Proposition 5.1]{TV06}, imply the assertion.
\end{proof}
Later on we give an alternative construction of this evolution in terms of integrable correlation functions.

\subsection{Local evolution of quasi-observables}
For technical reasons, such as uniqueness and convergence $\delta \to 0$, 
it is convenient to study first solutions to the pre-dual Cauchy problem \eqref{CORR:00}.
For this purpose, introduce $\Lb_{\alpha} := L^1(\Gamma_0^2, e^{\alpha|\eta|}d \lambda)$ with norm
\[
 \Vert G \Vert_{\Lb_{\alpha}} = \int \limits_{\Gamma_0^2}|G(\eta)|e^{\alpha|\eta|}d \lambda(\eta)
\]
where $e^{\alpha|\eta|} = e^{\alpha^+|\eta^+|}e^{\alpha^-|\eta^-|}$. 
\begin{Remark}
 Note that $\| \cdot \|_{\alpha} \leq \| \cdot \|_{\beta}$ and $\Lb_{\beta} \subset \Lb_{\alpha}$ for all $\alpha^{\pm} \leq \beta^{\pm}$,
 i.e. $(\Lb_{\alpha})_{\alpha \in \R^2}$ is a decreasing two-parameter scale of Banach spaces with dense embeddings $\Lb_{\alpha} \subset \Lb_{\alpha'}$.
\end{Remark}
Let $\widehat{L}_{\delta,\e} = \widehat{L}_{\delta}^S + \frac{1}{\e}\widehat{L}_{\delta}^E$
where $\widehat{L}_{\delta}^S = A_{\delta} + B_{\delta}$ and 
\begin{align*}
 M(\eta) &=  (m + \lambda)|\eta^+| + \sum \limits_{x \in \eta^+}\sum \limits_{y \in \eta^+ \backslash x}a^-(x-y) + g\sum \limits_{x \in \eta^+}\sum \limits_{y \in \eta^-}b^-(x-y),
 \\ (A_{\delta}G)(\eta) &= - M(\eta)G(\eta) + \sum \limits_{x \in \eta^+}R_{\delta}(x)\int \limits_{\R^d}a^+(x-y)G(\eta^+ \cup y, \eta^-)d y,
 \\ (B_{\delta}G)(\eta) &= - \sum \limits_{x \in \eta^+}\sum \limits_{y \in \eta^+ \backslash x}a^-(x-y)G(\eta^+ \backslash x, \eta^-) - g \sum \limits_{x \in \eta^+}\sum \limits_{y \in \eta^-}b^-(x-y)G(\eta^+, \eta^- \backslash y)
 \\ &\ \ \ + \lambda|\eta^+|G(\eta) + \sum \limits_{x \in \eta^+}R_{\delta}(x)\int \limits_{\R^d}a^+(x-y)G(\eta^+ \backslash x \cup y, \eta^-)d y,
 \\ (\widehat{L}_{\delta}^EG)(\eta) &= - |\eta^-|G(\eta) + \sum \limits_{\xi^- \subset \eta^-}\int \limits_{\R^d}z_{\delta}(x)e^{-E_{\psi}(x,\xi^-)}\mathcal{E}\left( e^{-\psi(x-\cdot)} - 1;\eta^- \backslash \xi^-\right)G(\eta^+, \xi^- \cup x)d x,
\end{align*}
where $\mathcal{E}(f;\eta^-) = \prod_{x \in \eta^-}f(x)$.
It is not difficult to see that for all $\alpha^{\pm} < \beta^{\pm}$ the operator $\widehat{L}_{\delta,\e}$ is bounded from $\Lb_{\beta}$ to $\Lb_{\alpha}$.
Moreover, we have for all $\beta^+ > \alpha^+$ and $\beta^- \geq \alpha^-$
\begin{align}\label{EQ:20}
 \Vert B_{\delta}G \Vert_{\Lb_{\alpha}} \leq \frac{e^{\beta^+}\| a^- \|_{L^1} + \lambda + \| a^+ \|_{L^1} + e^{\beta^-}g}{e(\beta^+ - \alpha^+)}\Vert G \Vert_{\Lb_{\beta}}.
\end{align}
The relation of $\widehat{L}_{\delta,\e}$, $\widehat{L}_{\delta}^S$ and $\widehat{L}_{\delta}^E$ to $L_{\e}, L^S, L^E$ is explained in the following remark.
\begin{Remark}\label{REMARK:01}
 Using \cite{FKO13} one can show that
 \[
  K\widehat{L}_{\delta}^EG_1 = L_{\delta}^EKG_1, \ \ K\widehat{L}_{\delta}^SG_2 = L_{\delta}^SKG_2, \ \ \mathbbm{K}\widehat{L}_{\delta,\e}G_3 = L_{\delta,\e}\mathbbm{K}G_3
 \]
 holds for all $G_1 \in B_{bs}(\Gamma_0^-), G_2 \in B_{bs}(\Gamma_0^+)$ and $G_{3} \in B_{bs}(\Gamma_0^2)$.
\end{Remark}
For any $\alpha = (\alpha^+, \alpha^-) \in \R^2$ let
\begin{align}\label{EQ:38}
 \mathcal{D}_{\alpha} = \begin{cases} 
                           \{ G \in \Lb_{\alpha} \ | \ (|\eta^-| + M)\cdot G \in \Lb_{\alpha} \}, & \ \e \in (0, \infty)\\ 
                           \{ G \in \Lb_{\alpha} \ | \ M \cdot G \in \Lb_{\alpha} \}, & \ \e = \infty
                        \end{cases}.
\end{align}
Let $\alpha_*^+$ and $m$ be given by \eqref{EQ:13}.
\begin{Lemma}\label{LEMMA:00}
 For each $\delta \geq 0$, each $\alpha^+ > \alpha_*^+$ and $\alpha^- \in [\alpha_*^-, \alpha^{*,-})$
 the operator $(A_{\delta} + \frac{1}{\e}\widehat{L}_{\delta}^E, \mathcal{D}_{\alpha})$ is the generator of an analytic, 
 semigroup $S_{\delta,\e}^{\alpha}(t)$ of contractions on $\Lb_{\alpha}$.
 Moreover, given $\beta^+ > \alpha^+$ and $\beta^- \in [\alpha^-, \alpha^{*,-})$, then $S_{\delta,\e}^{\beta}(t) = S_{\delta,\e}^{\alpha}(t)|_{\Lb_{\beta}}$ for all $t \geq 0$.
\end{Lemma}
\begin{proof}
 We consider only the case where $\lambda > 0$. The other case follows by similar arguments.
 For each $0 \leq G \in \mathcal{D}_{\alpha}$ we get by \eqref{IBP}, $R_{\delta} \leq 1$ and \eqref{EQ:07}
 \begin{align*}
  &\ \int \limits_{\Gamma_0^2}\sum \limits_{x \in \eta^+}R_{\delta}(x)\int \limits_{\R^d}a^+(x-y)G(\eta^+ \cup y, \eta^-)d y e^{\alpha|\eta|}d \lambda(\eta)
  \\ &\leq \max\left\{ \frac{\lambda}{\lambda+m}, \vartheta \right\} e^{-\alpha^+} \int \limits_{\Gamma_0^2}M(\eta)G(\eta)e^{\alpha|\eta|}d \lambda(\eta).
 \end{align*}
 Similarly we obtain by $z_{\delta} \leq z$
 \begin{align*}
  &\ \int \limits_{\Gamma_0^2}\sum \limits_{\xi^- \subset \eta^-}\int \limits_{\R^d}z_{\delta}(x)e^{-E_{\psi}(x,\xi^-)}\mathcal{E}\left( |e^{-\psi(x - \cdot)} - 1|;\eta^- \backslash \xi^-\right) G(\eta^+, \xi^- \cup x)d xd \lambda(\eta)
  \\ &\leq z C_{\psi}(\alpha^-)e^{- \alpha^-} \int \limits_{\Gamma_0^2}G(\eta^+, \xi^-) |\xi^-| e^{\alpha^+|\eta^+|}e^{\alpha^-|\xi^-|}d \lambda(\eta^+, \xi^-).
 \end{align*}
 Since $\max\left\{ \frac{b}{b+m}, \vartheta \right\} e^{-\alpha^+} < 1$ and $z C_{\psi}(\alpha^-)e^{- \alpha^-} < 1$ 
 the assertion can be deduced by similar arguments to \cite[Proposition 3.1]{FK16b}.
\end{proof}
A pair $(\alpha,\beta)$ is said to be admissible if $\beta^+ > \alpha^+ > \alpha_*^+$ and $\alpha_*^- \leq \alpha^- \leq \beta^- < \alpha^{*,-}$.
For such a pair let
\[
 T(\alpha, \beta) := \frac{\beta^+ - \alpha^+}{e^{\beta^+}\| a^- \|_{L^1} + \lambda + \| a^+ \|_{L^1} + e^{\beta^-}g}.
\]
\begin{Proposition}\label{PROP:00}
 For each $\delta \geq 0$ there exists a family of bounded linear operators
 \[
  \left\{ \widehat{U}^{\beta,\alpha}_{\delta,\e}(t) \in L(\Lb_{\beta}, \Lb_{\alpha}) \ | \ 0 \leq t < T(\alpha,\beta), \ (\alpha,\beta) \text{ admissible pair } \right\}
 \]
 with 
 \begin{align}\label{EQ:18}
  \Vert \widehat{U}^{\beta,\alpha}_{\delta,\e}(t) \Vert_{L(\Lb_{\beta}, \Lb_{\alpha})} \leq \frac{T(\alpha,\beta)}{T(\alpha,\beta) - t}, \ \ 0 \leq t < T(\alpha,\beta)
 \end{align}
 such that the following properties are satisfied.
 \begin{enumerate}
  \item[(a)] For any admissible pair $(\alpha,\beta)$ with $\alpha^- < \beta^-$ and any $G \in \Lb_{\beta}$, 
  $G_t := \widehat{U}^{\beta,\alpha}_{\delta,\e}(t)G$ is the unique classical solution in $\Lb_{\alpha}$ to
  \begin{align}\label{EQ:08}
   \frac{\partial G_t}{\partial t} = \widehat{L}_{\delta,\e}G_t, \ \ G_t|_{t = 0}, \ \ t \in [0, T(\alpha,\beta)).
  \end{align}
  \item[(b)] Given two admissible pairs $(\alpha_0, \alpha)$ and $(\alpha,\beta)$, we have 
  \begin{align}\label{EQ:10}
   \widehat{U}_{\delta,\e}^{\beta,\alpha_0}(t)G = \widehat{U}^{\alpha,\alpha_0}_{\delta,\e}(t)G = \widehat{U}^{\beta, \alpha}_{\delta,\e}(t)G
  \end{align}
  for any $G \in \Lb_{\beta}$ and $0 \leq t < \min\{ T(\alpha,\beta), T(\alpha_0, \alpha), T(\alpha_0,\beta)\}$.
 \end{enumerate}
\end{Proposition}
\begin{proof}
 In view of Lemma \ref{LEMMA:00} and \eqref{EQ:20} the assertion follows from \cite{F16a}.
\end{proof}
By property \eqref{EQ:10} we simply write $\widehat{U}_{\delta,\e}(t)$ instead of $\widehat{U}^{\beta,\alpha}_{\delta,\e}(t)$ if no confusion may arise.

\subsection{Local evolution of correlation functions}
Here and below we let $\delta \geq 0$ be arbitrary.
In this section we study the adjoint Cauchy problem to \eqref{EQ:08}. 
The dual space $(\Lb_{\alpha})^*$ can be identified with $\K_{\alpha}$ by use of the duality
\begin{align}\label{EQ:11}
 \langle \langle G, k \rangle \rangle = \int \limits_{\Gamma_0^2}G(\eta)k(\eta)d \lambda(\eta), \ \ G \in \Lb_{\alpha}, \ \ k \in \K_{\alpha}.
\end{align}
\begin{Remark}
 For all $\alpha^{\pm} < \beta^{\pm}$ we have $\| \cdot \|_{\K_{\beta}} \leq \| \cdot \|_{\K_{\alpha}}, \K_{\alpha} \subset \K_{\beta}$.
 Hence $(\K_{\alpha})_{\alpha \in \R^2}$ is an increasing two-parameter scale of Banach spaces.
 Note that the embeddings $\K_{\alpha} \subset \K_{\beta}$ are not dense.
\end{Remark}
Since $\widehat{L}_{\delta,\e} \in L(\Lb_{\beta}, \Lb_{\alpha})$ for all $\alpha^{\pm} < \beta^{\pm}$,
the adjoint operator satisfies $L^{\Delta}_{\delta,\e} := \widehat{L}_{\delta,\e}^* \in L(\K_{\alpha}, \K_{\beta})$.
It is for all $G \in \Lb_{\beta}$ and $k \in \K_{\alpha}$ determined by the relation
\[
 \int \limits_{\Gamma_0^2}(\widehat{L}_{\delta,\e}G)(\eta)k(\eta)d \lambda(\eta) = \int \limits_{\Gamma_0^2}G(\eta) (L^{\Delta}_{\delta,\e}k)(\eta)d \lambda(\eta).
\]
Moreover, if $\mu \in \mathcal{P}_{\alpha}$, then by Remark \ref{REMARK:01} and definition of correlation functions
\[
 \int \limits_{\Gamma^2}(L_{\delta,\e}F)(\gamma)d\mu(\gamma) = \int \limits_{\Gamma_0^2}(\widehat{L}_{\delta,\e}G)(\eta)k_{\mu}(\eta)d \lambda(\eta), \ \ F = \mathbbm{K}G \in \mathcal{FP}(\Gamma^2).
\]
The action of this operator is given by
$L_{\delta,\e}^{\Delta} = L_{\delta}^{\Delta,S} + \frac{1}{\e}L_{\delta}^{\Delta,E}$ 
with $L_{\delta}^{\Delta,S} = A_{\delta}^{\Delta} + B_{\delta}^{\Delta}$ and
\begin{align}
  \label{EQ:35} (A_{\delta}^{\Delta}k)(\eta) &= - M(\eta)k(\eta) + \sum \limits_{x \in \eta^+}\sum \limits_{y \in \eta^+ \backslash x}R_{\delta}(y)a^+(x-y)k(\eta^+ \backslash x, \eta^-),
 \\ \label{EQ:36} (B_{\delta}^{\Delta}k)(\eta) &= - \sum \limits_{x \in \eta^+}\int \limits_{\R^d}a^-(x-y)k(\eta^+ \cup y, \eta^-)d y - g \sum \limits_{x \in \eta^+}\int \limits_{\R^d}b^-(x-y)k(\eta^+, \eta^- \cup y)d y
 \\ \notag &\ \ \ + \lambda|\eta^+|k(\eta) + \sum \limits_{x \in \eta^+}\int \limits_{\R^d}R_{\delta}(y) a^+(x-y)k(\eta^+ \backslash x \cup y, \eta^-)d y,
 \\ \label{EQ:37} (L_{\delta}^{\Delta, E}k)(\eta) &= -|\eta^-|k(\eta) 
 \\ \notag &\ \ \ + \sum \limits_{x \in \eta^-}z_{\delta}(x)e^{-E_{\psi}(x,\eta^- \backslash x)}\int \limits_{\Gamma_0^-}\mathcal{E}\left( e^{-\psi(x-\cdot)} - 1;\xi^-\right)k(\eta^+, \xi^- \cup \eta^- \backslash x)d \lambda(\xi^-).
\end{align}
The operator $B_{\delta}^{\Delta}$ satisfies for any $\alpha^{+} < \beta^{+}$ and $\alpha^- \leq \beta^-$
\[ 
 \Vert B_{\delta}^{\Delta}k \Vert_{\K_{\beta}} \leq \frac{e^{\beta^+}\| a^- \|_{L^1} + \lambda + \| a^+ \|_{L^1} + e^{\beta^-}g}{e(\beta^+ - \alpha^+)} \Vert k \Vert_{\K_{\alpha}}.
\]
Given any admissible pair $(\alpha,\beta)$ and $0 \leq t < T(\alpha,\beta)$ we let 
$U_{\delta,\e}^{\Delta, \alpha,\beta}(t) := \widehat{U}_{\delta,\e}^{\beta,\alpha}(t)^*$ be the adjoint operators. 
\begin{Proposition}\label{TH:00}
 Fix any $\delta \geq 0$. The family of adjoint operators
 \[
  \left \{ U_{\delta,\e}^{\Delta, \alpha,\beta}(t) \in L(\K_{\beta}, \K_{\alpha}) \ | \ 0 \leq t < T(\alpha,\beta), \ \ (\alpha,\beta) \text{ admissible pair } \right\}
 \]
 satisfies the following properties.
 \begin{enumerate}
  \item[(a)] Let $\alpha_*^+ < \alpha_0^+ < \alpha^+ < \beta^+$ and $\alpha_*^- \leq \alpha_0^- \leq \alpha^- \leq \beta^- < \alpha^{*,-}$. Then
   \begin{align}\label{EQ:33}
    U_{\delta,\e}^{\Delta,\alpha_0,\beta}(t)k = U_{\delta,\e}^{\Delta, \alpha_0, \alpha}(t)k = U_{\delta,\e}^{\Delta,\alpha,\beta}(t)k
   \end{align}
   for any $k \in \K_{\alpha_0}$ and $0 \leq t < \min \{ T(\alpha,\beta), T(\alpha_0, \alpha), T(\alpha_0,\beta)\}$.
  \item[(b)] For each admissible pair $(\alpha,\beta)$ with $\alpha^- < \beta^-$ and each $k_0 \in \K_{\alpha}$, 
   $k_t := U_{\delta,\e}^{\Delta, \alpha,\beta}(t)k_0$ is the unique classical solution in $\K_{\beta}$ to 
  \begin{align}\label{EQ:12}
   \frac{\partial k_t}{\partial t} = L_{\delta,\e}^{\Delta}k_t,\ \ k_t|_{t=0} = k_0, \ \ 0 \leq t < T(\alpha,\beta).
  \end{align}
  Moreover, if $k_0 \in \K_{\alpha_0}$ with $\alpha_0^{\pm} < \alpha^{\pm}$, then 
  $U_{\delta,\e}^{\Delta, \alpha,\beta}(t) L_{\delta,\e}^{\Delta}k_0 = L_{\delta,\e}^{\Delta}U_{\delta,\e}^{\Delta, \alpha,\beta}(t)k_0$
  holds for all $t \in [0,T(\alpha,\beta))$.
  \item[(c)] Let $(k_t)_{t \in [0,T)} \subset \K_{\beta}$ be such that $[0,T)\ni t \longmapsto \langle \langle G, k_t \rangle \rangle$ is continuous for every $G \in \Lb_{\beta}$ and
  \begin{align}\label{EQ:09}
   \langle \langle G, k_t \rangle \rangle = \langle \langle G, k_0 \rangle \rangle + \int \limits_{0}^{t}\langle \langle \widehat{L}_{\delta,\e}G, k_s^{\e} \rangle \rangle d s, \ \ t \in [0,T), \ \ G \in B_{bs}(\Gamma_0^2).
  \end{align}
  Then $k_{t} = U_{\delta,\e}^{\Delta, \alpha,\beta}(t)k_0$ holds for all $0 \leq t < \min\{ T, T(\alpha,\beta)\}$.
 \end{enumerate} 
\end{Proposition}
\begin{proof}
 Follows from \cite{F16a}.
\end{proof}
Due to property \eqref{EQ:33} we omit the additional dependence on $(\alpha,\beta)$, if no confusion may arise. 
In such a case we let $U_{\delta,\e}^{\Delta}(t)$ stand for $U_{\delta,\e}^{\Delta,\alpha,\beta}(t)$. 
\begin{Proposition}
 Let $k_0 \in \K_{\alpha}$, $G \in B_{bs}(\Gamma_0^2)$ and $(\alpha,\beta)$ be any admissible pair with $\alpha^- < \beta^-$.
 Then for each $0 \leq t < \frac{1}{3}T(\alpha,\beta)$
 \begin{align}\label{EQ:17}
  \langle \langle G, U_{\delta,\e}^{\Delta}(t)k_0 \rangle \rangle \longrightarrow \langle \langle G, U_{0,\e}^{\Delta}(t)k_0 \rangle \rangle, \ \ \delta \to 0.
 \end{align}
\end{Proposition}
\begin{proof}
 Let $(\alpha,\beta)$ be any admissible pair with $\alpha^- < \beta^-$ and fix any $0 \leq t < \frac{1}{3}T(\alpha,\beta)$.
 Define $\beta_0^{\pm} := \frac{\alpha^{\pm} + \beta^{\pm}}{2}$. Then $\alpha^{\pm} < \beta_0^{\pm} < \beta^{\pm}$,
 $T(\alpha,\beta_0) > \frac{1}{2}T(\alpha,\beta)$ and $T(\beta_0, \beta) = \frac{1}{2}T(\alpha,\beta)$.
 Since $\beta_1 \longmapsto T(\beta_1, \beta)$ is continuous, we can find $\beta_1^{\pm} \in (\beta_0^{\pm}, \beta^{\pm})$ such that
 \begin{align}\label{EQ:25}
  \frac{1}{3}T(\alpha,\beta) \leq \min\{ T(\alpha,\beta_0), T(\beta_1, \beta)\}.
 \end{align}
 Hence for any $s \in [0,t]$ it follows that $s < T(\alpha,\beta_0)$ and $t-s < T(\beta_1, \beta)$.
 Consequently $U_{\delta,\e}^{\Delta}(s)k_0 \in \K_{\beta_0}$ and by \eqref{EQ:33} $U_{0,\e}^{\Delta}(t-s)k_0 \in \K_{\beta}$ for $s \in [0,t]$. 
 By \eqref{EQ:12} we see that
 \[
  [0,t] \ni s \longmapsto U_{0,\e}^{\Delta}(t-s)U_{\delta,\e}^{\Delta}(s)k_0 \in \K_{\beta}
 \]
 is continuously differentiable in $\K_{\beta}$ such that
 \[
  U_{\delta,\e}^{\Delta}(t)k_0 - U_{0,\e}^{\Delta}(t)k_0 = \int \limits_{0}^{t}U_{0,\e}^{\Delta}(t-s)\left( L_{\delta,\e}^{\Delta} - L_{0,\e}^{\Delta}\right) U_{\delta,\e}^{\Delta}(s)k_0 d s
 \]
 holds in $\K_{\beta}$.
 Take $G \in B_{bs}(\Gamma_0^2) \subset \Lb_{\beta}$, let $G_{t-s}^{\e} := \widehat{U}_{0,\e}(t-s)G \in \Lb_{\beta_1}$
 and $k_s^{\delta,\e} := U_{\delta,\e}^{\Delta}(s)k_0 \in \K_{\beta_0}$. Then 
 \begin{align*}
   |(L_{\delta,\e}^{\Delta}k_s^{\delta,\e})(\eta) - (L_{0,\e}^{\Delta}k_s^{\delta,\e})(\eta)| \leq e^{\beta_0^+|\eta^+|}e^{\beta_0^-|\eta^-|} H_{\delta,\e, \beta_0}(\eta) \Vert k_s^{\delta,\e} \Vert_{\K_{\beta_0}}
 \end{align*}
 where 
 \begin{align*}
  H_{\delta,\e,\beta_0}(\eta) &:= e^{-\beta_0^+}\sum \limits_{x \in \eta^+}\sum \limits_{y \in \eta^+ \backslash x}a^+(x-y)(1 - R_{\delta}(y))
   + \sum \limits_{x \in \eta^+}\int \limits_{\R^d}(1 - R_{\delta}(y))a^+(x-y) d y
  \\ &\ \ \ + e^{- \beta_0^-}z C_{\psi}(\beta_0^-)\sum \limits_{x \in \eta^-}(1 - R_{\delta}(x))e^{- E_{\psi}(x,\eta^- \backslash x)} \geq 0.
 \end{align*}
 This implies
 \begin{align*}
 |\langle \langle G, U_{\delta,\e}^{\Delta}(t)k_0 - U_{0,\e}^{\Delta}(t)k_0 \rangle \rangle |
 &\leq \int \limits_{0}^{t} |\langle \langle G, U_{0,\e}^{\Delta}(t-s)( L_{\delta,\e}^{\Delta} - L_{0,\e}^{\Delta}) k_{s}^{\delta,\e} \rangle \rangle| d s
 \\ &= \int \limits_{0}^{t} |\langle \langle G_{t-s}^{\e}, ( L_{\delta,\e}^{\Delta} - L_{0,\e}^{\Delta}) k_{s}^{\delta,\e} \rangle \rangle| d s
 \\ &\leq \sup \limits_{0 \leq r \leq t} \Vert k_{r}^{\delta,\e} \Vert_{\K_{\beta_0}}\int \limits_{0}^{t}\int \limits_{\Gamma_0^2}|G_{t-s}^{\e}(\eta)| H_{\delta,\e,\beta_0}(\eta) e^{\beta_0 |\eta|} d \lambda(\eta)d s.
 \end{align*}
 By duality, \eqref{EQ:18} and \eqref{EQ:25} we get 
 $\Vert k_{r}^{\delta,\e} \Vert_{\K_{\beta_0}} \leq \frac{T(\alpha,\beta_0)}{T(\alpha,\beta_0) - r} \leq \frac{3}{2}$, $r \in [0,t]$.
 Moreover $H_{\delta,\e,\beta_0}(\eta) \longrightarrow 0$ as $\delta \to 0$ and using
 \begin{align*}
  H_{\delta,\e,\beta_0}(\eta) \leq e^{-\beta_0^+}2 |\eta^+|^2 \Vert a^+ \Vert_{\infty} + 2 \| a^+ \|_{L^1} |\eta^+|
  + ze^{- \beta_0^-}C_{\psi}(\beta_0^-)|\eta^-| =: h(\eta)
 \end{align*}
 together with
 \[
  |\eta^{\pm}|^n e^{\beta_0^{\pm}|\eta^{\pm}|} \leq \frac{n^n}{(\beta_1^{\pm} - \beta_0^{\pm})^n} e^{-n} e^{\beta_1^{\pm}|\eta^{\pm}|}, \ \ n \geq 1, \ \eta^{\pm} \in \Gamma_0^{\pm}
 \]
 and $G_{t-s}^{\e} \in \Lb_{\beta_1}$ gives $G_{t-s}^{\e}\cdot h \cdot e^{\beta_0 |\cdot|} \in L^1(\Gamma_0^2, d\lambda)$.
 The assertion follows by dominated convergence.
\end{proof}

\subsection{Local evolution of states}
The main step in the proof is to show that $U_{0,\e}^{\Delta}(t)$ preserves positivity in the sense of Lennard, i.e.
\[
 \langle \langle G, U_{0,\e}^{\Delta}(t)k \rangle \rangle \geq 0, \ \ \mathbbm{K}G \geq 0, \ \ G \in B_{bs}(\Gamma_0^2)
\]
whenever $k$ is positive definite in the sense of Ruelle.
In view of \eqref{EQ:17}, we show that $U_{\delta,\e}^{\Delta}(t)$ for $\delta > 0$ enjoys such property. 
To this end we prove that $U_{\delta,\e}^{\Delta}(t)k_0$ is the correlation function of a probability density given by Lemma \ref{PROP:01}.

For this purpose we introduce certain auxiliary Banach spaces of integrable correlation functions.
Similar spaces have been used in \cite{KK16, FK16b, F17WR}.
Let $\mathcal{R}^{\delta}_{\alpha}$ be the Banach space of all equivalence classes of functions $u$ with norm
\[
 \Vert u \Vert_{\mathcal{R}_{\alpha}^{\delta}} = \esssup \limits_{\eta \in \Gamma_0^2}\ \frac{|u(\eta)| e^{- \alpha^+ |\eta^+|}e^{- \alpha^- |\eta^-|}}{\mathcal{E}(R_{\delta};\eta^+) \mathcal{E}(R_{\delta};\eta^-)}.
\]
It is the dual space to $\mathcal{B}_{\alpha}^{\delta} := L^1(\Gamma_0^2, e^{\alpha|\cdot|}\mathcal{E}(R_{\delta})d \lambda)$ with norm
\[
 \Vert G \Vert_{\mathcal{B}_{\alpha}^{\delta}} = \int \limits_{\Gamma_0^2}|G(\eta)|e^{\alpha^+|\eta^+|}e^{\alpha^-|\eta^-|}\mathcal{E}(R_{\delta};\eta^+)\mathcal{E}(R_{\delta};\eta^-)d \lambda(\eta)
\]
and duality given by \eqref{EQ:11}.
Note that $\mathcal{R}_{\alpha}^{\delta} \subset \K_{\alpha}$, $\Lb_{\alpha} \subset \mathcal{B}_{\alpha}^{\delta}$,
$\mathcal{R}_{\alpha}^{\delta} \subset \bigcap_{\beta}\Lb_{\beta}$
and $\bigcup_{\beta}\K_{\beta} \subset \mathcal{B}_{\alpha}^{\delta}$.
The following is an analogue of Proposition \ref{PROP:00} and Proposition \ref{TH:00} with $\K_{\alpha}, \Lb_{\alpha}$ replaced by $\mathcal{B}_{\alpha}^{\delta}$
and $\mathcal{R}_{\alpha}^{\delta}$, respectively.
\begin{Proposition}
 For any $\delta > 0$ there exists a family of bounded linear operators 
 \[
  \left \{ W_{\delta,\e}^{\Delta, \alpha,\beta}(t) \in  L(\mathcal{R}_{\alpha}^{\delta}, \mathcal{R}_{\beta}^{\delta}) \ | \ 0 \leq t < T(\alpha,\beta), \ \ (\alpha,\beta) \ \text{ admissible pair } \right\}
 \]
 with 
 \begin{align}\label{EQ:21}
  \Vert W_{\delta,\e}^{\Delta,\alpha,\beta}(t) \Vert_{L(\mathcal{R}_{\alpha}^{\delta}, \mathcal{R}_{\beta}^{\delta})} \leq \frac{T(\alpha,\beta)}{T(\alpha,\beta) - t}, \ \ 0 \leq t < T(\alpha,\beta)
 \end{align}
 such that the following holds.
 \begin{enumerate}
  \item[(a)] Let $(\alpha_0, \alpha)$ and $(\alpha,\beta)$ be two admissible pairs. Then 
  \[
   W_{\delta,\e}^{\Delta, \alpha_0,\beta}(t)u = W_{\delta,\e}^{\Delta, \alpha_0, \alpha}(t)u = W_{\delta,\e}^{\Delta,\alpha,\beta}(t)u
  \]
  for any $u \in \K_{\alpha_0}$ and $0 \leq t < \min \{ T(\alpha_0, \alpha), T(\alpha,\beta), T(\alpha_0,\beta)\}$.
  \item[(b)] For each admissible pair $(\alpha,\beta)$ and $u_0 \in \mathcal{R}^{\delta}_{\alpha}$, 
   $u_t^{\delta,\e} := W_{\delta,\e}^{\Delta, \alpha,\beta}(t)u_0$ is the unique classical solution in $\mathcal{R}^{\delta}_{\beta}$ to
  \begin{align}\label{EQ:14}
   \frac{\partial u_t^{\delta,\e}}{\partial t} = L_{\delta,\e}^{\Delta}u_{t}^{\delta,\e}, \ \ 0 \leq t < T(\alpha,\beta)
  \end{align}
  Moreover, if $u_0 \in \mathcal{R}^{\delta}_{\alpha_0}$ with $\alpha_0^{\pm} < \alpha^{\pm}$, then 
  $W_{\delta,\e}^{\Delta, \alpha,\beta}(t) L_{\delta,\e}^{\Delta}u_0 = L_{\delta,\e}^{\Delta}W_{\delta,\e}^{\Delta, \alpha,\beta}(t)u_0$.
  \item[(c)] For any $u_0 \in \mathcal{R}_{\alpha}^{\delta}$ we have
  \begin{align}\label{EQ:16}
   U_{\e, \delta}^{\Delta,\alpha, \beta}(t)u_0 = W_{\delta,\e}^{\Delta, \alpha,\beta}(t)u_0
  \end{align}
  where $0 \leq t < T(\alpha,\beta)$ and $(\alpha,\beta)$ is any admissible pair.
 \end{enumerate} 
\end{Proposition}
\begin{proof}
 Observe that $\widehat{L}_{\delta,\e}$ is a bounded linear operator from $\mathcal{B}_{\beta}^{\delta}$ to $\mathcal{B}_{\alpha}^{\delta}$
 for any $\alpha^{\pm} < \beta^{\pm}$. Moreover, we have
 \[
  \Vert B_{\delta}G \Vert_{\mathcal{B}_{\alpha}^{\delta}} \leq \frac{e^{\beta^+}\| a^- \|_{L^1} + \lambda + \| a^+ \|_{L^1} + e^{\beta^-}g }{e (\beta^+ - \alpha^+)}\Vert G \Vert_{\mathcal{B}_{\beta}^{\delta}}
 \]
 for any $G \in \mathcal{B}_{\beta}^{\delta}$. 
 By \cite{F16a} we get a family of operator $\widehat{W}_{\delta,\e}^{\beta,\alpha}(t) \in L(\mathcal{B}_{\beta}^{\delta}, \mathcal{B}_{\alpha}^{\delta})$ 
 with similar properties as stated in Proposition \ref{PROP:00}.
 Then the adjoint operators $W_{\delta,\e}^{\Delta,\alpha,\beta}(t) := \widehat{W}_{\delta,\e}^{\beta,\alpha}(t)^*$ satisfy properties (a), (b).
 Let us prove (c). From \eqref{EQ:14} we see that
 $W_{\delta,\e}^{\Delta, \alpha,\beta}(t)u_0 \in \mathcal{R}_{\beta}^{\delta} \subset \K_{\beta}$ satisfies \eqref{EQ:09}
 for any $G \in B_{bs}(\Gamma_0^2)$. 
 Moreover, $t \longmapsto \langle \langle G, W_{\delta,\e}^{\Delta, \alpha,\beta}(t)u_0 \rangle \rangle$ is continuous for any $G \in \Lb_{\beta} \subset \mathcal{B}_{\beta}^{\delta}$.
 The assertion follows from Proposition \ref{TH:00}.(c).
\end{proof}
\begin{Remark}
 The family $W_{\delta,\e}^{\Delta}(t)$ satisfies a similar uniqueness property as given in Proposition \ref{PROP:00}.(c).
 However, it will not be used in this work.
\end{Remark}
Here and below we omit the superscript $(\alpha,\beta)$ if no confusion may arise.
\begin{Lemma}\label{TH:01}
 Let $(\alpha,\beta)$ be any admissible pair and $u_0 \in \mathcal{R}_{\alpha}^{\Delta}$ be positive definite.
 Then $W_{\delta,\e}^{\Delta}(t)u_0$ is positive definite for any $0 \leq t < T(\alpha,\beta)$.
\end{Lemma}
\begin{proof}
 For any $a = (a^+,a^-), b = (b^+,b^-) \in \R^2$ and $u \in \mathcal{R}_{a}^{\delta}$ let
 \[
  \mathcal{H}u(\eta) := \int \limits_{\Gamma_0^2}(-1)^{|\xi|}u(\eta \cup \xi)d \lambda(\xi), \ \ \eta \in \Gamma_0^2.
 \]
 Then $\mathcal{H}: \mathcal{R}_{a}^{\delta} \longrightarrow \Lb_{b}$ is continuous with
 \begin{align*}
  \Vert \mathcal{H}u \Vert_{\Lb_{b}} 
  \leq \Vert u \Vert_{\mathcal{R}_{a}^{\delta}} \exp\left( \left( e^{a^+ + b^+} + e^{a^- + b^-} + e^{a^+} + e^{a^-}\right) \| R_{\delta} \|_{L^1} \right).
 \end{align*}
 It is easily seen that $\mathcal{H}u \in \mathcal{D}_{\delta,\e} \subset D(\mathcal{J}_{\delta,\e})$ holds for
 for any $a \in \R^2$ and $u \in \mathcal{R}_{a}^{\delta}$.
 Moreover, for each $G \in \K_{c}$ and $u \in \mathcal{R}_{a}^{\delta}$ with $c = (c^+, c^-) \in \R_+^2$ one can show that
 \begin{align}\label{EQ:15}
  \langle \langle G, u \rangle \rangle = \langle \langle \mathbbm{K} G, \mathcal{H}u \rangle \rangle.
 \end{align}
 Since $L_{\delta,\e}^{\Delta}$ is bounded from $\mathcal{R}_{a}^{\delta}$ to $\mathcal{R}_{a'}^{\delta}$ for any $a'^{\pm} > a^{\pm}$ we get
 \begin{align*}
  \langle \langle \mathbbm{K}G, \mathcal{H}L_{\delta,\e}^{\Delta}u \rangle \rangle &= \langle \langle G, L_{\delta,\e}^{\Delta}u \rangle \rangle 
  = \langle \langle \widehat{L}_{\delta,\e}G, u \rangle \rangle
  \\ &= \langle \langle \mathbbm{K} \widehat{L}_{\delta,\e}G, \mathcal{H}u \rangle \rangle
  = \langle \langle L_{\delta,\e} \mathbbm{K}G, \mathcal{H}u \rangle \rangle = \langle \langle \mathbbm{K}G, \mathcal{J}_{\delta,\e}\mathcal{H}u \rangle \rangle
 \end{align*}
 where the last equality follows from \eqref{EQ:06} and $\mathcal{H}u \in \mathcal{D}_{\delta,\e}$.
 As $G$ was arbitrary we get $\mathcal{H}L_{\e,\delta}^{\Delta}u = \mathcal{J}_{\delta,\e}\mathcal{H}u$.
 Let $h_t^{\delta, \e} := \mathcal{H}W_{\delta,\e}^{\Delta}(t)u_0$ for $0 \leq t < T(\alpha,\beta)$.
 Then $h_t^{\delta, \e} \in L^1(\Gamma_0^2,d\lambda)$ and it is not difficult to see that 
 $D_{\delta,\e} \cdot h_t^{\delta, \e} \in L^1(\Gamma_0^2,d\lambda)$, i.e. $h_t^{\delta,\e} \in \mathcal{D}_{\delta,\e} \subset D(\mathcal{J}_{\delta,\e})$. 
 By \eqref{EQ:14} we get
 \[
  \frac{\partial h_t^{\delta,\e}}{\partial t} = \mathcal{J}_{\delta,\e}h_t^{\delta,\e}, \ \ h_t^{\delta,\e}|_{t=0} = \mathcal{H}u_0, \ \ t \in [0, T(\alpha,\beta))
 \]
 in $L^1(\Gamma_0^2, d \lambda)$. Denote by $\mathcal{S}_{\delta,\e}(t)$ the strongly continuous semigroup generated 
 $(\mathcal{J}_{\delta,\e}, D(\mathcal{J}_{\delta,\e}))$. Fix $t \in (0,T(\alpha,\beta))$, then
 $[0,t] \ni s \longmapsto \mathcal{S}_{\delta,\e}(t-s)h_s^{\delta,\e} \in L^1(\Gamma_0^2, d \lambda)$
 is continuously differentiable with $\frac{d}{d s}\mathcal{S}_{\delta,\e}(t-s)h_s^{\delta,\e} = 0$. 
 This implies $h_t^{\delta,\e} = \mathcal{S}_{\delta,\e}(t)\mathcal{H}u_0 \geq 0$, since by \eqref{EQ:15} we have $\mathcal{H}u_0 \geq 0$.
 But this implies $\langle\langle G, W_{\delta,\e}^{\Delta}(t)u_0 \rangle\rangle = \langle \langle \mathbbm{K}G, h_t^{\delta,\e} \rangle\rangle \geq 0$
 for any $G \in B_{bs}(\Gamma_0^2)$ with $\mathbbm{K}G \geq 0$.
\end{proof}
The next proposition follows by \eqref{EQ:17} and Lemma \ref{TH:01}.
\begin{Proposition}\label{TH:07}
 Let $(\alpha,\beta)$ be any admissible pair with $\alpha^- < \beta^-$ and let $k_0 \in \K_{\alpha}$ be positive definite.
 Then $U_{0,\e}^{\Delta}(t)k_0$ is positive definite for any $0 \leq t < \frac{1}{3}T(\alpha,\beta)$.
\end{Proposition}
\begin{proof}
 Fix $\delta > 0$ and let $k_{0,\delta}(\eta) := \mathcal{E}(R_{\delta};\eta^+)\mathcal{E}(R_{\delta};\eta^-)k_0(\eta)$. 
 Then $k_{0,\delta} \in \mathcal{R}_{\alpha}^{\delta'}$ for any $\delta' \in (0,\delta]$ and in view of \cite[Lemma 10]{F17WR} it is positive definite in the 
 sense of Lennard. Moreover, $W_{\delta',\e}^{\Delta}(t)k_{0,\delta}$ is positive definite for $0 \leq t < T(\alpha,\beta)$.
 Let $G \in B_{bs}(\Gamma_0^2)$ be such that $\mathbbm{K}G \geq 0$.
 By \eqref{EQ:16}, Lemma \ref{TH:01} and \eqref{EQ:17} we get for any $0 \leq t < \frac{1}{3}T(\alpha,\beta)$
 \[
  0 \leq \langle \langle G, W_{\delta',\e}^{\Delta}(t)k_{0,\delta} \rangle \rangle = \langle \langle G, U_{\delta',\e}^{\Delta}(t)k_{0,\delta}\rangle \rangle \longrightarrow \langle \langle G, U_{0,\e}^{\Delta}(t)k_{0,\delta} \rangle \rangle, \ \ \delta' \to 0.
 \]
 This shows that $U_{0,\e}^{\Delta}(t)k_{0,\delta}$ is positive definite. Letting $\delta \to 0$ yields the assertion.
\end{proof}
Existence and uniqueness of (local) solutions to the Fokker-Planck equation \eqref{EQ:05} is stated below.
\begin{Proposition}
 Let $(\alpha,\beta)$ be an admissible pair with $\alpha^- < \beta^-$ and $\mu_0 \in \mathcal{P}_{\alpha}$.
 The following assertions hold.
 \begin{enumerate}
  \item[(a)] There exists a family of states $(\mu^{\e}_t)_{0 \leq t < \frac{1}{3}T(\alpha,\beta)} \subset \mathcal{P}_{\beta}$ 
  such that for each $F \in \mathcal{FP}(\Gamma^2)$ we have $L_{0,\e}F \in L^1(\Gamma^2,d\mu_t^{\e})$,
  $[0,\frac{1}{3}T(\alpha,\beta)) \ni t \longmapsto \langle L_{0,\e}F, \mu_t^{\e}\rangle$ is continuous and \eqref{EQ:05}
  holds for all such $t$. Moreover, the correlation functions satisfy
  \begin{align}\label{EQ:22}
   \Vert k_{\mu_t^{\e}} \Vert_{\K_{\beta}} \leq \frac{3}{2}, \ \ 0 \leq t < \frac{1}{3}T(\alpha,\beta).
  \end{align}
  \item[(b)] Let $(\nu_t^{\e})_{0 \leq t < T} \subset \mathcal{P}_{\beta}$ be another family of states such that 
   for each $F \in \mathcal{FP}(\Gamma^2)$ we have $L_{0,\e}F \in L^1(\Gamma^2,d\nu_t^{\e})$,
  $[0,\frac{1}{3}T(\alpha,\beta)) \ni t \longmapsto \langle L_{0,\e}F, \nu_t^{\e} \rangle$ is locally integrable and \eqref{EQ:05}
  holds for a.a. $t$. Suppose that
  \begin{align}\label{EQ:24}
   \sup \limits_{t \in [0,T')} \Vert k_{\nu_t^{\e}} \Vert_{\K_{\beta}} < \infty, \ \ \forall T' \in (0,T).
  \end{align}
  Then $\nu_t^{\e} = \mu_t^{\e}$ for all $0 \leq t < \min \left\{T, \frac{1}{3}T(\alpha,\beta)\right\}$.
 \end{enumerate}
\end{Proposition}
\begin{proof}
 (a) Previous proposition shows that $U_{0,\e}^{\Delta}(t)k_0$ is positive definite.
 Hence there exists a family of states $(\mu_t^{\e})_{0 \leq t < \frac{1}{3}T(\alpha,\beta)} \subset \mathcal{P}_{\beta}$
 such that $U_{0,\e}^{\Delta}(t)k_0 = k_{\mu_t^{\e}}$ (see \cite{L74,L75}).
 Let $F \in \mathcal{FP}(\Gamma^2)$ and $G \in B_{bs}(\Gamma_0^2)$ with $F = \mathbbm{K}G$. Take $\beta_0^{\pm} > \beta^{\pm}$, then $G \in \Lb_{\beta_0}$
 and $\widehat{L}_{0,\e}G \in \Lb_{\beta}$. For each $0 \leq t < \frac{1}{3}T(\alpha,\beta)$ we have
 \[
  \langle \langle G, k_{\mu_t^{\e}} \rangle \rangle =  \langle \langle G, k_{0} \rangle \rangle + \int \limits_{0}^{t} \langle \langle \widehat{L}_{0,\e}G, k_{\mu_s^{\e}}\rangle \rangle ds.
 \]
 Since $\Lb_{\beta} \subset L^1(\Gamma_0^2, k_{\mu_s^{\e}}d\lambda)$ for $s \in [0,t]$ and 
 $\mathbbm{K}: L^1(\Gamma_0^2, k_{\mu_s^{\e}}d\lambda) \longrightarrow L^1(\Gamma^2, d\mu_s^{\e})$ is continuous (see \cite{KK02}) it follows that 
 $\mathbbm{K}\widehat{L}_{0,\e}G \in L^1(\Gamma^2, d\mu_s^{\e})$. By $\mathbbm{K} \widehat{L}_{0,\e}G = L_{0,\e}F$ and \eqref{EQ:30} we see that
 $t \longmapsto \langle \langle \widehat{L}_{0,\e}G, k_{\mu_t^{\e}}\rangle \rangle$ is continuous and
 \eqref{EQ:05} holds for $0 \leq t \frac{1}{3}T(\alpha,\beta)$.
 This shows that $(\mu_t^{\e})_{0 \leq t < \frac{1}{3}T(\alpha,\beta)}$ has the desired properties where \eqref{EQ:22} follows by duality, 
 \eqref{EQ:18} and $t < \frac{1}{3}T(\alpha,\beta)$.
 \newline 
 (b) Let $G \in B_{bs}(\Gamma_0^2)$. Then \eqref{EQ:05} implies that
 $[0,T) \ni t \longmapsto \langle \mathbbm{K}G,\nu_t^{\e} \rangle = \langle \langle G, k_{\nu_t^{\e}} \rangle \rangle$ is continuous and satisfies
 \begin{align*}
  \langle \langle G, k_{\nu_t^{\e}} \rangle \rangle = \langle \langle G, k_{\mu_0} \rangle \rangle + \int \limits_{0}^{t}\langle \langle \widehat{L}_{\e}G, k_{\nu_s^{\e}} \rangle \rangle ds.
 \end{align*}
 By Proposition \ref{TH:00}.(c) we get $k_{\nu_t^{\e}} = k_{\mu_t^{\e}}$ for all $0 \leq t < \min\{ T, T(\alpha,\beta) \}$ and hence $\mu_t^{\e} = \nu_t^{\e}$,
 provided we can show that $t \longmapsto \langle \langle G, k_{\nu_t^{\e}} \rangle \rangle$ is continuous for any $G \in \Lb_{\beta}$.
 This clearly holds for all $G \in B_{bs}(\Gamma_0^2)$. The general case $G \in \Lb_{\beta}$ follows by \eqref{EQ:24} and a standard density argument.
\end{proof}

\subsection{Global evolution of states}
It remains to extend the constructed evolution of states to all $t \geq 0$. Following the ideas from \cite{KK16}, 
we establish reasonable a priori estimates on solutions to \eqref{EQ:09}. Such estimates are used for the continuation of the local evolution of states.
Since we only deal with $\delta = 0$ below, we let $U_{\e}^{\Delta}(t) := U_{0,\e}^{\Delta}(t)$ and likewise 
$A^{\Delta} := A_0^{\Delta}, B^{\Delta} := B_0^{\Delta}, L^{\Delta, E} := L_0^{\Delta,E}, \widehat{L}_{\e} := \widehat{L}_{0,\e}$, etc.
Define a new operator $B_0^{\Delta}$ by
\[
 (B_0^{\Delta}k)(\eta) = \lambda|\eta^+|k(\eta) + \sum \limits_{x \in \eta^+}\int \limits_{\R^d} a^+(x-y)k(\eta^+ \backslash x \cup y, \eta^-)d y.
\]
For all $\alpha^+ < \beta^+$ with $\alpha^- = \beta^-$ we get
\[
 \Vert B_0^{\Delta}k\Vert_{\K_{\beta}} \leq \frac{\lambda + \| a^+ \|_{L^1}}{e(\beta^+ - \alpha^+)}\Vert k \Vert_{\K_{\alpha}}.
\]
Let $T_1(\alpha,\beta) := \frac{\beta^+ - \alpha^+}{\lambda + \| a^+ \|_{L^1}} \geq T(\alpha,\beta)$.
\begin{Lemma}\label{TH:02}
 There exists a family of bounded linear operators
 \[
  \left\{ Q_{\e}^{\Delta, \alpha,\beta}(t) \in L(\K_{\alpha}, \K_{\beta}) \ | \ 0 \leq t < T_1(\alpha,\beta), \ \ \log(\vartheta) < \alpha^+ < \beta^+, \ \ \alpha_*^- \leq \alpha^- \leq \beta^- < \alpha^{*,-} \right\}
 \]
 with $\Vert Q_{\e}^{\Delta,\alpha,\beta}(t) \Vert_{L(\K_{\alpha},\K_{\beta})} \leq \frac{T_1(\alpha,\beta)}{T_1(\alpha,\beta) - t}$, $0 \leq t < T_1(\alpha,\beta)$
 such that for all $k \in \K_{\alpha}$ and 
 \[
  \log(\vartheta) < \alpha^+ < \beta^+, \ \alpha_*^- \leq \alpha^- \leq \beta^- < \alpha^{*,-}
 \]
 the following properties hold.
 \begin{enumerate}
  \item[(a)] Let $\alpha_0^{+} \in (\alpha^+, \beta^+)$ and $\alpha_0^- \in [\alpha^-, \beta^-]$. Then 
  \[
   Q_{\e}^{\Delta,\alpha,\alpha_0}(t)k = Q_{\e}^{\Delta,\alpha_0,\beta}k = Q_{\e}^{\Delta, \alpha,\beta}(t)k
  \]
  for any $0 \leq t < \min\{ T_1(\alpha,\alpha_0), T_1(\alpha_0,\beta), T_1(\alpha,\beta)\}$. Moreover
  \[
   \left(A^{\Delta} + \frac{1}{\e}L^{\Delta, E} + B_0^{\Delta}\right)Q_{\e}^{\Delta,\alpha,\alpha_0}(t)k = Q_{\e}^{\Delta,\alpha_0,\beta}(t)\left(A^{\Delta} + \frac{1}{\e}L^{\Delta, E} + B_0^{\Delta}\right)k
  \]
  holds for $0 \leq t < \min \{ T_1(\alpha,\alpha_0), T_1(\alpha_0,\beta)\}$.
  \item[(b)] If $k \geq 0$, then $Q_{\e}^{\Delta, \alpha,\beta}(t)k \geq 0$ for $t \in [0, T_1(\alpha,\beta))$.
  \item[(c)] Suppose that $\alpha^- < \beta^-$, then $k_t^{\e} := Q_{\e}^{\Delta,\alpha,\beta}(t)k$ is the unique classical solution in $\K_{\beta}$ to
  \[
   \frac{\partial k_t^{\e}}{\partial t} = \left(A^{\Delta} + \frac{1}{\e}L^{\Delta, E} + B_0^{\Delta}\right)k_t^{\e}, \ \ 0 \leq t < T_1(\alpha,\beta).
  \]
  \item[(d)] Let $k$ be positive definite. Then $U_{\e}^{\Delta,\alpha,\beta}(t)k \leq Q_{\e}^{\Delta,\alpha,\beta}(t)k$ for $0 \leq t < \frac{1}{3}T(\alpha,\beta)$.
 \end{enumerate} 
\end{Lemma}
\begin{proof}
 Existence of such family together with property (a) follows by \cite[Theorem 3.1]{F16}. 
 \newline (b) The construction of $Q_{\e}^{\Delta}(t)$ (see \cite[Theorem 3.1]{F16}) shows that it can be obtained by an approximation
 build by positive operators.
 \newline (c) This follows from \cite{F16a} together with the observation that $A^{\Delta} + \frac{1}{\e}L^{\Delta, E} + B_0^{\Delta}$ is a bounded linear
 operator from $\K_{\alpha}$ to $\K_{\beta}$ for any $\alpha^{\pm} < \beta^{\pm}$.
 \newline (d) Take $\beta_0,\beta$ as in the proof of Proposition \ref{TH:00}.(c). Then we can find $\beta_1^{\pm} \in (\beta_0^{\pm}, \beta^{\pm})$ with \eqref{EQ:25}.
 Hence $[0,t] \ni s \longmapsto Q_{\e}^{\Delta,\beta_1,\beta}(t-s)U_{\e}^{\Delta,\alpha,\beta_0}(s)k_0 \in \K_{\beta}$ is continuously differentiable which implies
 \[
  Q_{\e}^{\Delta,\alpha,\beta}(t)k_0 - U_{\e}^{\Delta,\alpha,\beta}(t)k_0 
  = \int \limits_{0}^{t}Q_{\e}^{\Delta,\beta_1,\beta}(t-s)\left( B_0^{\Delta} - B^{\Delta}\right)U_{\e}^{\Delta,\alpha,\beta_0}(s)k_0d s
 \]
 for $0 \leq t < \frac{1}{3}T(\alpha,\beta)$.
 Since $U_{\e}^{\Delta,\alpha,\beta_0}(s)k_0 \geq 0$ (see Corollary \ref{TH:07}) and $Q_{\e}^{\Delta,\beta_1,\beta}(t-s)$ preserves positivity, 
 the assertion follows from $\left( B_0^{\Delta} - B^{\Delta}\right)U_{\e}^{\Delta,\alpha,\beta_0}(s)k_0 \geq 0$.
\end{proof}
Due to part (a) we omit the dependence on $(\alpha,\beta)$.
The next statement provides a global correlation function evolution.
\begin{Proposition}
 Let $\alpha_0^+ > \alpha_*^+$, $\alpha_*^- \leq \alpha^-_0 < \alpha^{*,-}$ and $k_0 \in \K_{\alpha_0}$ be the correlation function
 for some $\mu_0 \in \mathcal{P}_{\alpha_0}$. Then there exists $k_t \subset \bigcup_{\beta^+ > \alpha_0^+}\K_{\beta^+, \alpha_0^-}, \ t\geq 0$,
 with the following properties.
 \begin{enumerate}
  \item[(a)] $k_t$ is positive definite for any $t \geq 0$.
  \item[(b)] For all $G \in B_{bs}(\Gamma_0^2)$
  \[
   \langle \langle G, k_t \rangle \rangle = \langle \langle G, k_0 \rangle \rangle + \int \limits_{0}^{t}\langle \langle \widehat{L}_{\e}G, k_s \rangle \rangle d s, \ \ t \geq 0.
  \]

  \item[(c)] Let $(r_t)_{t \in [0,T)} \subset \K_{\beta^+, \alpha_0^-}$ for some $\beta^+ > \alpha_0^+$ be such that
  $t \longmapsto \langle \langle G, r_t \rangle \rangle$ is continuous for any $G \in \Lb_{\beta^+, \alpha_0^-}$ and
  \[
   \langle \langle G, r_t \rangle \rangle = \langle \langle G, k_0 \rangle \rangle + \int \limits_{0}^{t}\langle \langle \widehat{L}_{\e}G, r_s \rangle \rangle d s, \ \ t \in [0,T), \ \ G \in B_{bs}(\Gamma_0^2).
  \]
  Then $k_t = r_t$ for all $0 \leq t < \min\{ T, T_1(\alpha_0, \beta)\}$ where $\beta = (\beta^+, \alpha_0^-)$.
 \end{enumerate}
\end{Proposition}
\begin{proof}
 First, we show that $U_{\e}^{\Delta}(t)k_{\mu_0}$ can be extended to all $t \geq 0$. Let $\kappa := \frac{\alpha^{*,-} - \alpha_0^-}{2}$,
 \begin{align}
   \label{EQ:03} \beta^+(\alpha^+, T) &:= \alpha^+ + (\lambda + \| a^+ \|_{L^1})T
  \\ \notag \tau(\alpha) &:= T(\alpha, (\alpha^+ + 1, \alpha^- + \kappa)) = \left( g e^{\kappa}e^{\alpha^-} + e^{\alpha^+}\| a^- \|_{L^1} e + \lambda + \| a^+ \|_{L^1} \right)^{-1}.
 \end{align}
 Define two sequences $(T^{(n)})_{n \in \N_0}$ and $(\alpha_n)_{n \in \N_0}$ by $T^{(0)} := 0$ and
 \begin{align}\label{EQ:04}
  T^{(n+1)} := \frac{1}{4}\tau(\alpha_n), \ \ \alpha_{n+1} := (\alpha_{n+1}^+, \alpha_{n+1}^-) := (\beta^+(\alpha_n^+,T^{(n+1)}), \alpha_0^-).
 \end{align}
 Let $k_t^{(1)} := U_{\e}^{\Delta}(t)k_{\mu_0}$ for $t \in [0,T^{(1)}]$ and define recursively
 \begin{align}\label{EQ:34}
  k_t^{(n+1)} := U_{\e}^{\Delta}(t)k^{(n)}_{T^{(n)}}, \ \ t \in [0, T^{(n+1)}]
 \end{align}
 where $n \geq 1$. Let us show the following properties:
 \begin{enumerate}
  \item[(i)] $k_t^{(n)} \in \K_{\alpha_n}$ for all $n \geq 1$ and $t \in [0, T^{(n)}]$.
  \item[(ii)] $\alpha_n^+ = \alpha_0^+ + (\| a^+ \|_{L^1} + \lambda)(T^{(1)} + \dots + T^{(n)})$ and $\alpha_n^- = \alpha_0^-$.
  \item[(iii)] $\sum_{n=1}^{\infty}T^{(n)} = \infty$.
 \end{enumerate}
 (i) and (ii): We proceed by induction. For $n = 1$ we get
 \[
  T^{(1)} = \frac{1}{4}\tau(\alpha_0) < \frac{1}{3}T(\alpha_0, (\alpha_0^+ + 1, \alpha_0^- + \kappa))
 \]
 and hence $k_t^{(1)} \in \K_{\alpha_0^+ + 1, \alpha_0^- + \kappa}$ and $k_t^{(1)} \leq Q_{\e}^{\Delta}(t)k_{\mu_0}$ for $t \in [0,T^{(1)}]$, see Lemma \ref{TH:02}.(d). 
 Property (ii) follows from
 \[
  \alpha_1^+ = \beta^+(\alpha_0^+, T^{(1)}) = \alpha_0^+ + (\| a^+ \|_{L^1} + \lambda)T^{(1)}.
 \]
 Since $Q_{\e}^{\Delta}(t)k_{\mu_0} \in \K_{\alpha_1}$ we obtain $k_t^{(1)} \in \K_{\alpha_1}$ for all $t \in [0,T^{(1)}]$.
 Suppose that (i) and (ii) hold for some $n \geq 1$. Then, by (ii), 
 \[
  \alpha_{n+1}^+ = \beta^+(\alpha_n^+, T^{(n+1)}) = \alpha_n^+ + (\lambda + \| a^+ \|_{L^1})T^{(n+1)} = \alpha_0^+ + (\lambda + \| a^+ \|_{L^1})(T^{(1)} + \dots + T^{(n+1)}).
 \]
 Finally, 
 \[
  k_t^{(n+1)} = U_{\e}^{\Delta}(t)k_{T^{(n)}}^{(n)} \leq Q_{\e}^{\Delta}(t)k_{T^{(n)}}^{(n)} \in \K_{\alpha_{n+1}}
 \]
 for $t \in [0,T^{(n+1)}]$, since $T^{(n+1)} < \frac{1}{3}T(\alpha_n, (\alpha_n^+ +1, \alpha_n^- + \kappa))$.
 
 Let us prove (iii). Suppose that $\sum_{n=1}^{\infty}T^{(n)} < \infty$. Then $T^{(n)} \longrightarrow 0$ as $n \to \infty$.
 From (ii) we get $\alpha_{n}^+ - \alpha_{n-1}^+ = ( \lambda + \| a^+ \|_{L^1}) T^{(n)}, \ \ n \geq 1$ and hence
 \[
  \sum \limits_{n=1}^{N}T^{(n)} = \frac{1}{\lambda + \| a^+ \|_{L^1}} ( \alpha_N^+ - \alpha_0^+ ), \ \ N \geq 1.
 \]
 In particular, $\alpha_N^+$ has to be bounded, say $\alpha_N^+ \leq c$ for all $N \geq 1$. But this implies
 \begin{align*}
  T^{(n)} &= \frac{1}{4}\tau(\alpha_n) = \frac{1}{4} \frac{1}{e \| a^- \|_{L^1} e^{\alpha_n^+} + \lambda + \| a^+ \|_{L^1} + g e^{\kappa} e^{\alpha^-}} 
  \\ &\geq \frac{1}{4} \frac{1}{e \| a^- \|_{L^1} e^{c} + \lambda + \| a^+ \|_{L^1} + g e^{\kappa} e^{\alpha^-}} > 0,
 \end{align*}
 which contradicts to $T^{(n)} \longrightarrow 0$ as $n \to \infty$. Hence (iii) holds.
 
 Define $S_n := T^{(1)} + \dots + T^{(n)}, \ n \geq 1$ and $S_0 := 0$. Let 
 \begin{align}\label{EQ:46}
  k_t := k_{t - S_{n-1}}^{(n)}, \ \ \ S_{n-1} \leq t \leq S_{n}, \ \  n \geq 1.
 \end{align}
 Then $k_t$ is well-defined for any $t \geq 0$ and due to the semigroup property of $U_{\e}^{\Delta}(t)$ (see \cite{F16a}) 
 the evolution obtained by this procedure is independent of the particular choice of $T^{(n)}$. 
 It remains to prove (a) -- (c). Property (a) follows by Corollary \ref{TH:07}.
 \newline (b) By (ii) it suffices to show that for any $n \geq 1$ and any $t \in [0, S_n ]$
 \begin{align}\label{EQ:26}
  \langle \langle G, k_t \rangle \rangle = \langle \langle G, k_0 \rangle \rangle + \int \limits_{0}^{t} \langle \langle \widehat{L}_{\e}G, k_s \rangle \rangle ds, \ \ G \in B_{bs}(\Gamma_0^2).
 \end{align}
 First observe that by \eqref{EQ:34} we have $k_t = k_{t - S_n}^{(n+1)} = U_{\e}^{\Delta}(t - S_n)k_{T^{(n)}}^{(n)}$.
 Proposition \ref{TH:00}.(b) together with the substitution $s = r - S_n$ implies
 \begin{align*}
  \langle \langle G, k_t \rangle \rangle &= \langle \langle G, k_{T^{(n)}}^{(n)}\rangle \rangle + \int \limits_{0}^{t - S_n} \langle \langle \widehat{L}_{\e}G, U_{\e}^{\Delta}(s)k_{T^{(n)}}^{(n)} \rangle \rangle ds
  \\ &= \langle \langle G, k_{T^{(n)}}^{(n)} \rangle \rangle + \int \limits_{S_n}^{t}\langle \langle \widehat{L}_{\e}G, U_{\e}^{\Delta}(r - S_n)k_{T^{(n)}}^{(n)} \rangle \rangle dr.
 \end{align*}
 Taking into account that $k_{T^{(n)}}^{(n)} = U_{\e}^{\Delta}(0)k_{T^{(n)}}^{(n)} = k_0^{(n+1)} = k_{S_n}$ and
 $U_{\e}^{\Delta}(r - S_n)k_{T^{(n)}}^{(n)} = k_{r - S_n}^{(n+1)} = k_r$ yields for $n \geq 0$ and $t \in [S_n, S_{n+1}]$
 \begin{align}\label{EQ:27}
  \langle \langle G, k_t \rangle \rangle = \langle \langle G, k_{S_n} \rangle \rangle + \int \limits_{S_n}^{t} \langle \langle \widehat{L}_{\e}G, k_s \rangle \rangle ds, \ \ G \in B_{bs}(\Gamma_0^2).
 \end{align}
 For $n = 0$ identity \eqref{EQ:26} readily follows from \eqref{EQ:27}. 
 Proceeding by induction we get for $t \in [S_n, S_{n+1}]$ by \eqref{EQ:27} and induction hypothesis 
 \begin{align*}
  \langle \langle G, k_t \rangle \rangle &= \langle \langle G, k_{S_n} \rangle \rangle + \int \limits_{S_n}^{t}\langle \langle \widehat{L}_{\e}G, k_s \rangle \rangle ds
  \\ &= \langle \langle G, k_0 \rangle \rangle + \int \limits_{0}^{S_n} \langle \langle \widehat{L}_{\e}G, k_s \rangle \rangle ds + \int \limits_{S_n}^{t}\langle \langle \widehat{L}_{\e}G, k_s \rangle \rangle ds
      = \langle \langle G, k_0 \rangle \rangle + \int \limits_{0}^{t} \langle \langle \widehat{L}_{\e}G, k_s \rangle \rangle ds.
 \end{align*}
  Finally, property (c) follows from the construction of $k_t$ and Proposition \ref{TH:00}.(c)
\end{proof}
Note that from \eqref{EQ:34} and \eqref{EQ:46} we obtain
\begin{align}\label{CORRDEF}
 k_t = k_{t- S_n}^{(n+1)} = U_{\e}^{\Delta}(t-S_n)k_{T^{(n)}}^{(n)}, \ \ t \in [S_n, S_{n+1}]
\end{align}
and $k_{T^{(n)}}^{(n)} = k_{S_n}$ for all $n \geq 0$. This completes, by Proposition \ref{TH:00}, the proof of Theorem \ref{TH:03}.

\section{Proof Theorem \ref{APRIORI}}

\subsection{Case $m \leq \| a^+ \|_{L^1}$}
Fix $\e > 0$. Let $\alpha^+ > \alpha_*^+$ and $\alpha^- \in [\alpha_*^-, \alpha^{*,-})$. 
Let $\delta$ and $\alpha_{\delta}^+$ be given as in Theorem \ref{APRIORI}.(a). Set 
\[
 r_t(\eta) := e^{\alpha_{\delta}^+|\eta^+|}e^{\alpha^-|\eta^-|}e^{\left( \| a^+ \|_{L^1} + \delta - m\right)|\eta^+|t}.
\]
Then $r_t \in K_{\alpha_{\delta}(t)}$ where $\alpha_{\delta}(t) = (\alpha_{\delta}^+ + ( \| a^+ \|_{L^1} + \delta - m)t, \alpha^-)$.
\begin{Lemma}
 For each $n \geq 0$ it holds that
 \[
  Q_{\e}(t)r_{S_n} \leq r_{S_n + t}, \ \ t \in [0,T^{(n+1)}]
 \]
 where $T^{(n+1)}$ is defined in \eqref{EQ:04} and $S_n := T^{(1)} + \dots + T^{(n)}$ with $S_0 := 0$.
\end{Lemma}
\begin{proof}
 Observe that $\delta - m < \lambda$ and $\alpha^+ \leq \alpha_{\delta}^+$ implies
 \[
  \alpha^+ + (\| a^+ \|_{L^1} + \delta - m)T \leq \beta^+((\alpha_{\delta}^+,\alpha^-), T), \ \ \forall T > 0.
 \]
 Hence we have $r_{S_n + t} \in \K_{\beta((\alpha_{\delta}^+, \alpha^-), S_{n+1})}$, $t \in [0,T^{(n+1)}]$. 
 By Lemma \ref{TH:02}.(a),(c) we get
 \begin{align}\label{EQ:49}
  Q_{\e}^{\Delta}(t)r_{S_n} - r_{S_n + t} = \int \limits_{0}^{t}Q_{\e}^{\Delta}(t-s)N_s ds
 \end{align}
 where $N_s \in \K_{\beta((\alpha_{\delta}^+, \alpha^-), S_{n+1}) + (\delta,\delta)}$ is given by
 \begin{align*}
  N_s(\eta) &= -\delta|\eta^+| r_{S_n +s }(\eta) + \frac{1}{\e}(L^{\Delta,E}r_{S_n + s})(\eta)
  \\ &\ \ \ - \sum \limits_{x \in \eta^+}\sum \limits_{y \in \eta^+ \backslash x}a^-(x-y)r_{S_n + s}(\eta) - g \sum \limits_{x \in \eta^+}\sum \limits_{y \in \eta^-}b^-(x-y)r_{S_n+s}(\eta)
  \\ &\ \ \ + \sum \limits_{x \in \eta^+}\sum \limits_{y \in \eta^+ \backslash x}a^+(x-y)r_{S_n +s}(\eta^+ \backslash x, \eta^-).
 \end{align*}
 By $L^{\Delta,E}r_{S_n+s} \leq 0$, condition (S), $e^{-\alpha_{\delta}^+}\vartheta < 1$ and $\lambda e^{- \alpha_{\delta}^+} < \delta$ we obtain
 \begin{align*}
  N_s(\eta) &\leq r_{S_n+s}(\eta)|\eta^+|\left( \lambda e^{- \alpha_{\delta}^+} - \delta\right)
  + r_{S_n + s}(\eta) \left( e^{- \alpha_{\delta}^+}\vartheta - 1\right)\sum \limits_{x \in \eta^+}\sum \limits_{y \in \eta^+ \backslash x}a^-(x-y) \leq 0.
 \end{align*}
 The assertion follows by Lemma \ref{TH:02}.(b).
\end{proof}
We prove Theorem \ref{APRIORI} by induction. For $n = 0$ and $t \in [0, T^{(1)}]$ we get by Lemma \ref{TH:02}.(d) and \eqref{CORRDEF}
\begin{align*}
 k_{\mu_t^{\e}} = U_{\e}^{\Delta}(t)k_{\mu_0} \leq Q_{\e}^{\Delta}(t)k_{\mu_0}.
\end{align*}
Next, observe that $k_{\mu_0} \leq \Vert k_{\mu_0}\Vert_{\K_{\alpha}}r_{0}$ and hence by previous lemma
\[
 Q_{\e}^{\Delta}(t)k_{\mu_0} \leq \Vert k_{\mu_0} \Vert_{\K_{\alpha}} Q_{\e}^{\Delta}(t)r_{0} \leq \Vert k_{\mu_0} \Vert_{\K_{\alpha}} r_{t}.
\]
This proves the assertion for $n = 0$. Suppose that the assertion holds for $n \geq 0$. 
Then $k_{T^{(n)}}^{(n)} = k_{S_n} \leq \Vert k_{\mu_0} \Vert_{\K_{\alpha}} r_{S_n}$ and hence by Lemma \ref{TH:02}.(d) and \eqref{CORRDEF}
\begin{align*}
 k_{\mu_t^{\e}} &= U_{\e}^{\Delta}(t)k_{T^{(n)}}^{(n)} \leq Q_{\e}^{\Delta}(t-S_n)k_{T^{(n)}}^{(n)}
 \leq \Vert k_{\mu_0} \Vert_{\K_{\alpha}} Q_{\e}^{\Delta}(t-S_n)r_{S_n} \leq \Vert k_{\mu_0}\Vert_{\K_{\alpha}} r_t.
\end{align*}
This proves the assertion in this case.

\subsection{Case $m > \| a^+ \|_{L^1}$}
Fix $\e > 0$. Let $\alpha^+ > \alpha_*^+$ and $\alpha^- \in [\alpha_*^-, \alpha^{*,-})$. 
Take $\delta$ and $\alpha_{\delta}^+$ as in Theorem \ref{APRIORI}. Set 
\[
 r_t(\eta) := e^{\alpha_{\delta}^+|\eta^+|}e^{\alpha^-|\eta^-|}e^{- \delta t}.
\]
Then $r_t \in K_{(\alpha_{\delta}^+, \alpha^-)}$ for all $t \geq 0$.
\begin{Lemma}
 For each $n \geq 0$ it holds that
 \[
  Q_{\e}(t)r_{S_n} \leq r_{S_n + t}, \ \ t \in [0,T^{(n+1)}]
 \]
 where $T^{(n+1)}$ is defined in \eqref{EQ:04} and $S_n := T^{(1)} + \dots + T^{(n)}$ with $S_0 := 0$.
\end{Lemma}
\begin{proof}
 By Lemma \ref{TH:02}.(a),(c) we see that \eqref{EQ:49} holds for $N_s \in \K_{(\alpha_{\delta}^+, \alpha^-) + (\delta,\delta)}$ given by
 \begin{align*}
  N_s(\eta) &= (- m + \| a^+ \|_{L^1})|\eta^+| r_{S_n + s}(\eta) + \frac{1}{\e}(L^{\Delta,E}r_{S_n + s})(\eta)
  \\ &\ \ \ - \sum \limits_{x \in \eta^+}\sum \limits_{y \in \eta^+ \backslash x}a^-(x-y)r_{S_n+s}(\eta) - g \sum \limits_{x \in \eta^+}\sum \limits_{y \in \eta^-}b^-(x-y)r_{S_n+s}(\eta)
  \\ &\ \ \ + \sum \limits_{x \in \eta^+}\sum \limits_{y \in \eta^+ \backslash x}a^+(x-y)r_{S_n+s}(\eta^+ \backslash x,\eta^-) + \delta r_{S_n + s}(\eta).
 \end{align*}
 Then $N_s(\emptyset,\eta^-) \leq 0$, 
 \[
  N_s(\{x\}, \eta^-) \leq r_{S_n+s}(\{x\},\eta^-)\left( - m + \| a^+ \|_{L^1} + \delta \right) \leq 0
 \]
 and for $|\eta^+| \geq 2$
 \begin{align*}
  N_s(\eta) &\leq r_{S_n +s}(\eta)\left( e^{- \alpha_{\delta}^+}\vartheta - 1 \right) \sum \limits_{x \in \eta^+}\sum \limits_{y \in \eta^+ \backslash x}a^-(x-y)
  \\ &\ \ \ + r_{S_n + s}(\eta)\left( \delta + |\eta^+|\left( \lambda e^{-\alpha_{\delta}^+} + \| a^+ \|_{L^1} - m \right) \right)
  \\ &\leq r_{S_n + s}(\eta)\delta(1 - |\eta^+|) \leq 0.
 \end{align*}
 The assertion follows by Lemma \ref{TH:02}.(b).
\end{proof}
The assertion now follows again by induction similarly to the previous case.

\section{Proof Theorem \ref{LYP:TH:01}}
Let $\Gamma_{\infty}^2 = \{ \gamma \in \Gamma^2 \ | \ \mathbb{M}(\gamma) < \infty \}$, where
\begin{align*}
  \mathbb{M}(\gamma) := \mathbb{V}(\gamma) 
  &+ \sum \limits_{x \in \gamma^+}\left( \sum\limits_{w \in \gamma^+ \backslash x}a^-(x-w)\right) \left(\sum \limits_{y \in \gamma^+ \backslash x}\Xi(x,y) + \sum \limits_{y \in \gamma^-}\Xi(x,y)\right)
  \\ &+ \sum \limits_{x \in \gamma^+}\left( \sum \limits_{w \in \gamma^-}b^-(x-w)\right)\left( \sum \limits_{y \in \gamma^+ \backslash x}\Xi(x,y) + \sum \limits_{y \in \gamma^-}\Xi(x,y) \right) < \infty.
\end{align*}
\begin{Lemma}
 For any $\mu \in \mathcal{P}$ we have $\mu(\Gamma_{\infty}^2) = 1$.
\end{Lemma}
\begin{proof}
 It suffices to show that $\int_{\Gamma^2}\mathbb{M}(\gamma)d\mu(\gamma) < \infty$. 
 The latter one can be shown by direct computation using the fact that $e \in L^1(\R^d)$ and $\Xi \in L^1(\R^d \times \R^d)$.
\end{proof}
It is not difficult to prove that $(L_{\e}\mathbb{V})(\gamma)$ is well-defined for any $\gamma \in \Gamma_{\infty}^2$.
\begin{Lemma}
 There exists a constant $c_{\e} > 0$ such that
 \[
  (L_{\e}\mathbb{V})(\gamma) \leq c_{\e}\mathbb{V}(\gamma) + \frac{z}{\e}\| e\|_{L^1}, \ \ \gamma \in \Gamma_{\infty}^2.
 \]
\end{Lemma}
\begin{proof}
Write $\mathbb{V}(\gamma) = V_0(\gamma) + V_1^+(\gamma) + V_1^-(\gamma) + W(\gamma)$
with $V_0(\gamma) = \sum_{x \in \gamma^+}e(x) + \sum_{x \in \gamma^-}e(x)$,
$V_1^+(\gamma) = \frac{1}{2}\sum_{x \in \gamma^+}\sum_{y \in \gamma^+ \backslash x}\Xi(x,y)$ and
$V_1^-(\gamma) = \frac{1}{2}\sum_{x \in \gamma^-}\sum_{y \in \gamma^- \backslash x}\Xi(x,y)$. 
Let us estimate each term in 
\begin{align}\label{LYP:00}
 (L_{\e}\mathbb{V})(\gamma) &= (L_{\e}V_0)(\gamma) + (L_{\e}V_1^+)(\gamma) + (L_{\e}V_1^-)(\gamma) + (L_{\e}W)(\gamma).
\end{align}
separately. For the first term we get
\begin{align}\label{LYP:01}
 (L_{\e}V_0)(\gamma) \leq \left( \left \Vert \frac{a^+}{e} \right \Vert_{L^1} - m \right) \sum \limits_{x \in \gamma^+}e(x) - \frac{1}{\e}\sum \limits_{x \in \gamma^-}e(x)
    + \frac{z}{\e} \Vert e \Vert_{L^1}
\end{align}
where we have used $\int_{\R^d}a^+(x-y)e(y)dy \leq e(x) \left \Vert \frac{a^+}{e} \right \Vert_{L^1}$.
For the second term we get
\begin{align*}
 (L_{\e}V_1^+)(\gamma) &\leq \sum \limits_{x \in \gamma^+}\int \limits_{\R^d}a^+(x-y)\sum \limits_{w \in \gamma^+}\Xi(w,y)dy
 \\ &= \sum \limits_{x \in \gamma^+}\sum \limits_{w \in \gamma^+ \backslash x}\int \limits_{\R^d}a^+(x-y)\Xi(w,y)dy + \sum \limits_{x \in \gamma^+}\int \limits_{\R^d}a^+(x-y)\Xi(x,y)dy
     = I_1 + I_2.
\end{align*}
Then we obtain
\[
 I_2 = \sum \limits_{x \in \gamma^+}e(x) \int \limits_{\R^d}a^+(y) \frac{1 + |y|^{\kappa}}{|y|^{\kappa}}e(x-y)dy \leq c_{\kappa} \sum \limits_{x \in \gamma^+}e(x)
\]
where $c_{\kappa} := \int_{\R^d}a^+(y) \frac{1 + |y|^{\kappa}}{|y|^{\kappa}}dy < \infty$. For $I_1$ we have
\[
 I_1 \leq \sum \limits_{x \in \gamma^+}\sum \limits_{w \in \gamma^+ \backslash x} e(w)e(x) \int \limits_{\R^d}\frac{a^+(y)}{e(y)} \frac{1 + |y - (x- w)|^{\kappa}}{|y - (x-w)|^{\kappa}}dy
 \leq c' V_1^+(\gamma)
\]
where 
\[
 c' := \sup \limits_{w \in \R^d} \int \limits_{\R^d}\frac{a^+(y)}{e(y)} \frac{1 + |y - w|^{\kappa}}{|y - w|^{\kappa}}dy
 \leq \left \Vert \frac{a^+}{e} \right \Vert_{L^{\infty}} \int \limits_{\R^d}\frac{1}{|y|^{\kappa}}dy + 2 \left \Vert \frac{a^+}{e} \right \Vert_{L^1} < \infty.
\]
Altogether we see that
\begin{align}\label{LYP:02}
 (L_{\e}V_1^+)(\gamma) &\leq c' V_1^+(\gamma) + c_b \sum \limits_{x \in \gamma^+}e(x).
\end{align}
In the same way
\begin{align}\label{LYP:03}
 (L_{\e}V_1^-)(\gamma) &\leq - \frac{1}{\e} \sum \limits_{x \in \gamma^-}e(x) + c_- \frac{z}{\e} \sum \limits_{w \in \gamma^-}e(x)
\end{align}
where $c_- = \int_{\R^d}e(x)\frac{1 + |x|^{\kappa}}{|x|^{\kappa}}dx$. Finally we have
\begin{align*}
 (L_{\e}W)(\gamma) &\leq \sum \limits_{x \in \gamma^+}\sum \limits_{w \in \gamma^-}\int \limits_{\R^d}a^+(x-y)\Xi(y,w)dy
 \\ &\ \ \  - \frac{1}{\e}\sum \limits_{x \in \gamma^-}\sum \limits_{y \in \gamma^+}\Xi(x,y) + \frac{z}{\e}\sum \limits_{w \in \gamma^+}\int \limits_{\R^d}e^{-E_{\psi}(x,\gamma^-)}\Xi(w,x)dx.
\end{align*}
Using
\begin{align*}
 &\ \sum \limits_{x \in \gamma^+}\sum \limits_{w \in \gamma^-}\int \limits_{\R^d}a^+(x-y)\Xi(y,w)dy
  = \sum \limits_{x \in \gamma^+}\sum \limits_{w \in \gamma^-}e(w)\int \limits_{\R^d}a^+(x-y)e(y)\frac{1 + |y-w|^{\kappa}}{|y-w|^{\kappa}}dy
  \\ &\leq \sum \limits_{x \in \gamma^+}\sum \limits_{w \in \gamma^-}\Xi(x,w)\int \limits_{\R^d}\frac{a^+(y)}{e(y)}\frac{1 + |x-y-w|^{\kappa}}{|x-y-w|^{\kappa}}dy
      \leq c' W(\gamma)
\end{align*}
and
\begin{align*}
 \frac{z}{\e}\sum \limits_{w \in \gamma^+}\int \limits_{\R^d}e^{-E_{\psi}(x,\gamma^-)}\Xi(w,x)dx
 \leq \frac{z}{\e}\sum \limits_{w \in \gamma^+}e(w) \int \limits_{\R^d}e(x)\frac{1 + |w-x|^{\kappa}}{|w-x|^{\kappa}}dx
 \leq c_- \frac{z}{\e} \sum \limits_{w \in \gamma^+}e(w)
\end{align*}
where $c_- := \sup_{w \in \R^d}\int_{\R^d}e(x)\frac{1 + |w-x|^{\kappa}}{|w-x|^{\kappa}}dx < \infty$ we obtain
\begin{align}\label{LYP:04}
 (L_{\e}W)(\gamma) &\leq (c' - m - 1)W(\gamma) + c_- \frac{z}{\e}\sum \limits_{x \in \gamma^+}e(x).
\end{align}
The assertion follows from \eqref{LYP:00} -- \eqref{LYP:04}.
\end{proof}
Assertion Theorem \ref{LYP:TH:01}.(b) can be deduced from approximation of $\mathbb{V}$ by $\mathbb{V}_n \in \mathcal{FP}(\Gamma^2)$
together with Gronwall lemma.

\section{Proof Theorem \ref{TH:05}}
The proof is divided into three steps. First we consider an auxiliary problem and prove a preliminary version of the averaging principle.
Based on this auxiliary result we show the stochastic averaging principle for the evolution of quasi-observables.
Finally, Theorem \ref{TH:05} follows by standard duality arguments.

\subsection{Auxiliary convergence}
Since here and below we only consider the case $\delta = 0$ we let $\widehat{L}^S = \widehat{L}^S_{0}, \widehat{L}^E = \widehat{L}^E_0, A = A_0, B = B_0$,
where the latter operators have been defined in Section 4. Let 
\begin{align*}
 (C_{\e}G)(\eta) &= (AG)(\eta) - g (VG)(\eta) + \frac{1}{\e}(\widehat{L}^EG)(\eta)
 \\ (B'G)(\eta) &= (BG)(\eta) + g(VG)(\eta)
\end{align*}
where $A$, $B$ and $\widehat{L}^E$ have been defined in \eqref{EQ:35} -- \eqref{EQ:37} and
\[
 (VG)(\eta) := \sum \limits_{x \in \eta^+}\sum \limits_{y \in \eta^-}b^-(x-y)G(\eta^+, \eta^- \backslash y).
\]
These are well-defined linear operators on $\Lb_{\alpha}$ with domain $\mathcal{D}_{\alpha}$ defined by \eqref{EQ:38} for $\e \in (0,\infty)$.
Here and below we assume that \eqref{EQ:47} holds.
\begin{Lemma}
 Let $\alpha^- \in [\alpha_*^-, \alpha^{*,-})$ and $\alpha^+ > \widetilde{\alpha}_*^+$.
 Then $(C_{\e}, \mathcal{D}_{\alpha})$ is the generator of an analytic semigroup $(T_{\e}^{\alpha}(t))_{t \geq 0}$
 of contractions on $\Lb_{\alpha}$. Moreover $T_{\e}^{\alpha}(t)|_{\Lb_{\beta}} = T_{\e}^{\beta}(t)$ holds for all $\alpha^{\pm} < \beta^{\pm}$ with $\beta^- < \alpha^{*,-}$. 
\end{Lemma}
\begin{proof}
 For $0 \leq G \in \mathcal{D}_{\alpha}$ we have
 \begin{align*}
  \int \limits_{\Gamma_0^2}(VG)(\eta)e^{\alpha|\eta|}d\lambda(\eta)
  &= e^{\alpha^-} \int \limits_{\Gamma_0^2}|\eta^+|G(\eta)e^{\alpha|\eta|}d \lambda(\eta).
 \end{align*}
 The assertion follows by similar arguments to Lemma \ref{LEMMA:00}.
\end{proof}
Next we construct the limiting semigroup when $\e \to 0$. 
Similarly to $\Lb_{\alpha}$ we define $\Lb_{\alpha^+}$ as the Banach space of all equivalence classes of functions with finite norm
\[
 \| G \|_{\Lb_{\alpha^+}} = \int \limits_{\Gamma_0^+}|G(\eta^+)|e^{\alpha^+|\eta^+|}d\lambda(\eta^+).
\]
Then $(\Lb_{\alpha^+})_{\alpha^+ \in \R}$ is an decreasing scale of Banach spaces with dense embeddings 
$\Lb_{\beta^+} \subset \Lb_{\alpha^+}$ whenever $\alpha^+ < \beta^+$. 
Define $\widehat{\overline{L}} = \overline{A} + \overline{B}$ where
\begin{align*}
 (\overline{A}G)(\eta^+) &= - \overline{M}(\eta^+)G(\eta^+)
   + \sum \limits_{x \in \eta^+}\int \limits_{\R^d}a^+(x-y)G(\eta^+ \cup y) d y,
 \\ (\overline{B}G)(\eta^+) &= \lambda |\eta^+| G(\eta^+)
     - \sum \limits_{x \in \eta^+}\sum \limits_{y \in \eta^+ \backslash x}a^-(x-y)G(\eta^+ \backslash x)   
     + \sum \limits_{x \in \eta^+}\int \limits_{\R^d}a^+(x-y)G(\eta^+ \backslash x \cup y)d y,
 \\ \overline{M}(\eta) &= \left( m + \lambda + g \rho \right)|\eta^+| + \sum \limits_{x \in \eta^+}\sum \limits_{y \in \eta^+ \backslash x}a^-(x-y).
\end{align*}
It is a well-defined linear operator on $\Lb_{\alpha^+}$ with domain
$\overline{\mathcal{D}} = \{ G \in \Lb_{\alpha^+} \ | \ \overline{M}\cdot G \in \Lb_{\alpha^+}\}$.
\begin{Remark}
 One can show that $K\widehat{\overline{L}}G := \overline{L}KG$ holds for any $G \in B_{bs}(\Gamma_0^+)$.
\end{Remark}
\begin{Lemma}
 Let $\alpha^+ > \alpha_*^+$. Then $(\overline{A}, \overline{\mathcal{D}}_{\alpha^+})$ is the generator of an analytic, positive semigroup 
 $\overline{T}^{\alpha^+}(t)$ of contractions on $\Lb_{\alpha^+}$. 
 Moreover $\overline{T}^{\alpha^+}(t)|_{\Lb_{\beta^+}} = \overline{T}^{\beta^+}(t)$ holds for $\beta^+ > \alpha^+$.
\end{Lemma}
\begin{proof}
 The assertion follows by similar arguments to Lemma \ref{LEMMA:00}.
\end{proof}
For $G_1 \in \Lb_{\alpha^+}$ let $(G_1 \otimes 0^-)(\eta) = G_1(\eta^+)0^{|\eta^-|}$, $\eta \in \Gamma_0^2$, and define
\[
 \Lb_{\alpha^+}\otimes 0^- := \{ G_1 \otimes 0^-  \ | \ G_1 \in \Lb_{\alpha^+} \}.
\]
This defines an isometric isomorphism
\[
 P_+: \Lb_{\alpha^+} \longrightarrow \Lb_{\alpha^+} \otimes 0^-, \ \ (P_+G_1)(\eta) = G_1(\eta^+)0^{|\eta^-|} = (G_1 \otimes 0^-)(\eta)
\]
with inverse $(P_+^{-1}(G_1 \otimes 0^-))(\eta^+) = G_1(\eta^+)$.
The semigroup $\overline{T}^{\alpha}(t)$ can be naturally extended onto $\Lb_{\alpha^+}\otimes 0^-$
via $\overline{T}^{\alpha^+}(t) \otimes 0^- := P_+\overline{T}^{\alpha}(t)P_+^{-1}$, i.e. for $G_1 \otimes 0^- \in \Lb_{\alpha^+}\otimes 0^-$ set
\[
 (\overline{T}^{\alpha^+}(t) \otimes 0^-)(G_1 \otimes 0^-) := (\overline{T}^{\alpha^+}(t)G_1) \otimes 0^-.
\]
It has generator $(\overline{A} \otimes 0^-, \overline{\mathcal{D}}_{\alpha^+} \otimes 0^-)$, where
$(\overline{A}\otimes 0^-)(G_1 \otimes 0^-) = (\overline{A}G_1) \otimes 0^-$ and 
\[
 \overline{\mathcal{D}}_{\alpha^+} \otimes 0^- = \{ G_1 \otimes 0^- \ | \ G_1 \in \overline{\mathcal{D}}_{\alpha^+} \}.
\]
\begin{Proposition}
 Let $\alpha^+ > \widetilde{\alpha}_*^+$ and $\alpha^- \in [\alpha_*^-, \alpha^{*,-})$. Then for each $G \in \Lb_{\alpha^+}$
 \begin{align}\label{EQ:40}
  T_{\e}^{\alpha}(t)P_+G  \longrightarrow P_+\overline{T}^{\alpha^+}(t)G, \ \ \e \to 0
 \end{align}
 holds in $\Lb_{\alpha}$ uniformly on compacts w.r.t. $t \geq 0$.
\end{Proposition}
\begin{proof}
  Observe that the following holds.
 \begin{enumerate}
  \item[(i)] Similar arguments as given in the proof of Lemma \ref{LEMMA:00} show that   
  $(A - gV, \widetilde{\mathcal{D}}_{\alpha})$ is the generator of an analytic semigroup of contractions where
  $\widetilde{\mathcal{D}}_{\alpha} = \{ G \in \Lb_{\alpha} \ |\ M \cdot G \in \Lb_{\alpha} \}$.
  \item[(ii)] By \cite[Lemma 6, Proposition 6]{FK17} it follows that   
  $(\widehat{L}^E, D_{\alpha}(\widehat{L}^E))$ is the generator of an analytic semigroup $(\widehat{T}^E_{\alpha}(t))_{t \geq 0}$ 
  of contractions on $\Lb_{\alpha}$ where $D_{\alpha}(\widehat{L}^E) = \{ G \in \Lb_{\alpha} \ | \ |\eta^-| G \in \Lb_{\alpha} \}$.
  
  Let $k_{\mathrm{inv}} \in \K_{\alpha_*^-}$ be the correlation function of the unique Gibbs measure with activity $z$ and potential $\psi$.
  Then $\widehat{T}^E_{\alpha}(t)$ is ergodic with projection operator
  \[
   (PG)(\eta) = 0^{|\eta^-|} \int \limits_{\Gamma_0^-}G(\eta^+,\xi^-)k_{\mathrm{inv}}(\xi^-)d \lambda(\xi^-),
  \]
  i.e. there exist constants $c_0,c_1 > 0$ such that for any $G \in \Lb_{\alpha}$
  \[
   \Vert \widehat{T}_{\alpha}^{E}(t)G - PG \Vert_{\Lb_{\alpha}} \leq c_0 e^{-c_1 t} \Vert G - PG \Vert_{\Lb_{\alpha}}, \ \ t \geq 0.
  \]
  \item[(iii)] By \cite{F17WR} it follows that $B_{bs}(\Gamma_0^2)$ is a core for $\widehat{L}^E$. 
  Since $B_{bs}(\Gamma_0^2) \subset \widetilde{\mathcal{D}}_{\alpha} \cap D_{\alpha}(\widehat{L}^E)$, we see 
  that $\widetilde{\mathcal{D}}_{\alpha} \cap D_{\alpha}(\widehat{L}^E) = \mathcal{D}_{\alpha}$ is a core for the generator $\widehat{L}^E$.
 \end{enumerate}
 It is easily seen that $\widetilde{\mathcal{D}}_{\alpha}\cap (\Lb_{\alpha^+}\otimes 0^-) = \overline{\mathcal{D}}_{\alpha^+} \otimes 0^-$
 and $P(A-gV)G = (\overline{A}G_1) \otimes 0^-$ holds for $G = G_1 \otimes 0^- \in \overline{\mathcal{D}}_{\alpha^+} \otimes 0^-$.
 Hence $(P(A - gV), \widetilde{\mathcal{D}}_{\alpha} \cap (\Lb_{\alpha^+}\otimes 0^-))$ is 
 the generator of $\overline{T}^{\alpha}(t) \otimes 0^-$. The assertion now follows from \cite[Theorem 2.1]{KURTZ73}.
\end{proof}

\subsection{Convergence of quasi-observables}
Using the relation $\widehat{L}_{\e} = C_{\e} + B'$ we construct the evolution of quasi-observables similarly to Proposition \ref{PROP:00}. 
Namely, observe that
\[
 \Vert B'G \Vert_{\Lb_{\alpha}} \leq \frac{e^{\alpha^+} \| a^- \|_{L^1} + \lambda + \| a^+ \|_{L^1}}{e(\beta^+ - \alpha^+)} \Vert G \Vert_{\Lb_{\beta}}
\]
holds for any $\alpha^+ < \beta^+$ and $\alpha^- \leq \beta^-$. Let 
\begin{align}\label{EQ:39}
 T_0(\alpha,\beta) = \frac{\beta^+ - \alpha^+}{e^{\beta^+}\| a^- \|_{L^1} + \lambda + \| a^+ \|_{L^1}}
\end{align}
for such $(\alpha,\beta)$. Then $T(\alpha,\beta) \leq T_0(\alpha,\beta)$ holds for all admissible pairs $(\alpha,\beta)$.
\begin{Proposition}
 There exists a family of bounded linear operators
 \[
  \left\{ \widehat{V}^{\beta,\alpha}_{\e}(t) \in L(\Lb_{\beta}, \Lb_{\alpha}) \ | \ 0 \leq t < T_0(\alpha,\beta), \ (\alpha,\beta) \text{ admissible pair } \right\}
 \]
 with $\Vert \widehat{V}^{\beta,\alpha}_{\e}(t) \Vert_{L(\Lb_{\beta}, \Lb_{\alpha})} \leq \frac{T_0(\alpha,\beta)}{T_0(\alpha,\beta) - t}$, 
 $0 \leq t < T_0(\alpha,\beta)$ such that
 \begin{enumerate}
  \item[(a)] For any admissible pair $(\alpha,\beta)$ with $\alpha^- < \beta^-$ and $G \in \Lb_{\beta}$, 
  $G_t := \widehat{V}^{\beta,\alpha}_{\e}(t)G$ is the unique classical solution in $\Lb_{\alpha}$ to
  \begin{align*}
   \frac{\partial G_t}{\partial t} = \widehat{L}_{\e}G_t, \ \ G_t|_{t = 0}, \ \ t \in [0, T_0(\alpha,\beta)).
  \end{align*}
  \item[(b)] Given two admissible pairs $(\alpha_0, \alpha)$ and $(\alpha,\beta)$, we have 
  \begin{align*}
   \widehat{V}^{\alpha,\alpha_0}_{\e}(t)G = \widehat{V}^{\beta, \alpha}_{\e}(t)G = \widehat{V}_{\e}^{\beta, \alpha_0}(t)G
  \end{align*}
  for any $G \in \Lb_{\beta}$ and $0 \leq t < \min\{ T_0(\alpha,\beta), T_0(\alpha_0, \alpha), T_0(\alpha_0, \beta)\}$.
 \end{enumerate}
\end{Proposition}
As before, we omit the additional dependence on $(\alpha,\beta)$.
\begin{Lemma}\label{LEMMA:01}
 Let $(\alpha,\beta)$ be an admissible pair with $\alpha^- < \beta^-$.
 Then $\widehat{V}^{\beta,\alpha}_{\e}(t) = \widehat{U}^{\beta,\alpha}_{\e}(t)$ holds for all $0 \leq t < T(\alpha,\beta)$.
\end{Lemma}
\begin{proof}
 Let $G \in \Lb_{\beta}$. Then $G_t := \widehat{V}^{\beta,\alpha}_{\e}(t)$ is a classical solution to \eqref{EQ:08}
 when restricted to $[0, T(\alpha,\beta))$. The assertion follows by uniqueness.
\end{proof}
Below we construct the limiting evolution of quasi-observables. For this purpose observe that
\[
 \Vert \overline{B}G \Vert_{\Lb_{\alpha^+}} \leq \frac{e^{\alpha^+} \| a^- \|_{L^1} + \lambda + \| a^+ \|_{L^1}}{e(\beta^+ - \alpha^+)} \Vert G \Vert_{\Lb_{\beta^+}}
\]
holds for all $\alpha^+ < \beta^+$. By abuse of notation we let $T_0(\alpha^+, \beta^+) = T_0(\alpha,\beta)$,
since $T_0$ given by \eqref{EQ:39} is independent of $\alpha^-,\beta^-$.
\begin{Proposition}
 There exists a family of bounded linear operators
 \[
  \left\{ \overline{U}^{\beta^+,\alpha^+}(t) \in L(\Lb_{\beta^+}, \Lb_{\alpha^+}) \ | \ 0 \leq t < T_0(\alpha^+,\beta^+), \ \beta^+ > \alpha^+ > \alpha_*^+ \right\}
 \]
 with $\Vert \overline{U}^{\beta^+,\alpha^+}(t) \Vert_{L(\Lb_{\beta^+}, \Lb_{\alpha^+})} \leq \frac{T_0(\alpha^+,\beta^+)}{T_0(\alpha^+,\beta^+) - t}, \ \ 0 \leq t < T_0(\alpha^+,\beta^+)$
 such that
 \begin{enumerate}
  \item[(a)] For any $\beta^+ > \alpha^+ > \alpha_*^+$ and any $G \in \Lb_{\beta^+}$, 
  $G_t := \overline{U}^{\beta^+,\alpha^+}(t)G$ is the unique classical solution in $\Lb_{\alpha^+}$ to
  \begin{align*}
   \frac{\partial G_t}{\partial t} = \widehat{\overline{L}}G_t, \ \ G_t|_{t = 0}, \ \ t \in [0, T_0(\alpha^+,\beta^+)).
  \end{align*}
  \item[(b)] Given $\alpha_*^+ < \alpha_0^+ < \alpha^+ < \beta^+$, we have 
  \begin{align*}
   \overline{U}^{\alpha^+,\alpha_0^+}(t)G = \overline{U}^{\beta^+, \alpha^+}(t)G = \overline{U}^{\beta^+, \alpha_0^+}(t)G
  \end{align*}
  for any $G \in \Lb_{\beta^+}$ and $0 \leq t < \min\{ T_0(\alpha^+,\beta^+), T_0(\alpha_0^+, \alpha^+), T_0(\alpha_0^+, \beta^+)\}$.
 \end{enumerate}
\end{Proposition}
Here and below we omit the dependence on $\alpha$ for $T_{\e}^{\alpha}(t)$, if no confusion may arise.
The next proposition establishes a local stochastic averaging principle.
\begin{Proposition}
 Let $(\alpha,\beta)$ be an admissible pair with $\alpha^- < \beta^-$, $G \in \Lb_{\beta^+}$ and $T \in (0, T(\alpha,\beta))$.
 Then
 \begin{align}\label{EQ:45}
  \widehat{U}_{\e}^{\beta,\alpha}(t)P_+G \longrightarrow P_+\overline{U}^{\beta^+,\alpha^+}(t)G, \ \ \e \to 0
 \end{align}
 holds uniformly on $[0,T]$ in $\Lb_{\alpha}$
\end{Proposition}
\begin{proof}
 Let $H_0^{\e}(t) := T_{\e}(t)$ and define recursively
 \[
  H_{n+1}^{\e}(t) := \int \limits_{0}^{t}T_{\e}(t-s)B' H_n^{\e}(s)ds, \ \ n \geq 0.
 \]
 Then, clearly we have for $n \geq 1$
 \[
  H_n^{\e}(t) = \int \limits_{0}^{t}\cdots \int \limits_{0}^{t_{n-1}} T_{\e}(t-t_1)B' \cdots B' T_{\e}(t_{n-1} - t_n)dt_n \cdots dt_1.
 \]
 Take $\alpha_j^{\pm} = \alpha^{\pm} + \frac{\beta^{\pm} - \alpha^{\pm}}{n}j$, $j = 0, \dots, n$.
 Then $\alpha_{j+1}^{\pm} - \alpha_{j}^{\pm} = \frac{\beta^{+} - \alpha^{+}}{n}$ and hence
 \[
  \Vert B' \Vert_{L(\Lb_{\alpha_{j+1}}, \Lb_{\alpha_j})} \leq \frac{e^{\alpha_{j+1}}\| a^- \|_{L^1} + \lambda + \| a^+ \|_{L^1}}{\alpha_{j+1} - \alpha_j}
  \leq n \frac{e^{\beta^+}\| a^- \|_{L^1} + \lambda + \| a^+ \|_{L^1}}{\beta^{+} - \alpha^{+}}.
 \]
 This implies $H_n^{\e}(t) \in L(\Lb_{\beta}, \Lb_{\alpha})$ such that
 \begin{align}
  \notag \Vert H_n^{\e}(t)G\Vert_{\Lb_{\alpha}} &\leq \Vert B' \Vert_{L(\Lb_{\alpha_1}, \Lb_{\alpha_0})} \cdots \Vert B' \Vert_{L(\Lb_{\alpha_n}, \Lb_{\alpha_{n-1}})}\Vert G \Vert_{\Lb_{\beta}}\frac{t^n}{n!}
  \\ \notag &\leq \frac{t^n}{n!}\frac{n^n}{e^n(\beta^+ - \alpha^+)^n} \left( e^{\beta^+}\| a^- \|_{L^1} + \lambda + \| a^+ \|_{L^1} \right)^n \Vert G \Vert_{\Lb_{\beta}}.
  \\ \label{EQ:41} &= \left( \frac{t}{T_0(\alpha^+, \beta^+)}\right)^n \Vert G \Vert_{\Lb_{\beta}}.
 \end{align}
 It follows that $\sum_{n=1}^{\infty}H_n^{\e}(t)$ converges uniformly w.r.t. the norm in $L(\Lb_{\beta}, \Lb_{\alpha})$ in $t \in [0,T]$.
 By \cite{F16a}, see also Lemma \ref{LEMMA:01}, we obtain
 \[
  \widehat{U}_{\e}^{\beta,\alpha}(t) = \widehat{V}_{\e}^{\beta,\alpha}(t) = T_{\e}(t) + \sum \limits_{n=1}^{\infty}H_n^{\e}(t).
 \]
 Let $\overline{H}_0(t) := \overline{T}(t)$ and define recursively
 \[
  \overline{H}_{n+1}(t) := \int \limits_{0}^{t}\overline{T}(t-s)\overline{B}\overline{H}_{n}(s)ds, \ \ n \geq 0.
 \]
 Similar arguments show that
 \begin{align}\label{EQ:42}
  \Vert \overline{H}_n(t)G \Vert_{\Lb_{\alpha^+}} \leq \left( \frac{t}{T_0(\alpha^+, \beta^+)} \right)^n \Vert G \Vert_{\Lb_{\beta^+}}
 \end{align}
 and hence by \cite{F16a} we see that
 \[
  \overline{U}^{\beta^+,\alpha^+}(t) = \overline{T}(t) + \sum \limits_{n=1}^{\infty}\overline{H}_n(t)
 \]
 converges uniformly w.r.t. the norm in $L(\Lb_{\beta^+}, \Lb_{\alpha^+})$ in $t \in [0,T]$. 
 Let $G_1 \in \Lb_{\beta^+}$ be arbitrary and take $N \in \N$. Then
 \begin{align*}
  &\ \Vert \widehat{U}_{\e}^{\beta,\alpha}(t)P_+G_1 - P_+\overline{U}^{\beta^+, \alpha^+}(t)G_1 \Vert_{\Lb_{\alpha}}
  \leq \Vert T_{\e}(t)P_+G_1 - P_+\overline{T}(t)G_1 \Vert_{\Lb_{\alpha}} 
  \\ &\ \ \ + \sum \limits_{n=1}^{N}\Vert H_n^{\e}(t)P_+G_1 - P_+\overline{H}_n(t)G_1 \Vert_{\Lb_{\alpha}}
  + \sum \limits_{n = N+1}^{\infty} \Vert H_n^{\e}(t)P_+G_1 \Vert_{\Lb_{\alpha}} + \sum \limits_{n=N+1}^{\infty} \Vert P_+\overline{H}_n(t)G_1 \Vert_{\Lb_{\alpha}}.
 \end{align*}
 The first term tends by \eqref{EQ:40} to zero uniformly in $t \in [0,T]$. 
 The last two terms tend by \eqref{EQ:41} and \eqref{EQ:42} uniformly in $t, \e$ to zero as $N \to \infty$.
 Thus it suffices to show that for each $n \geq 1$
 \begin{align}\label{EQ:43}
  \sup \limits_{t \in [0,T]}\Vert H_n^{\e}(t)P_+G_1 - P_+\overline{H}_n(t)G_1 \Vert_{\Lb_{\alpha}} \longrightarrow 0, \ \ \e \to 0.
 \end{align}
 Observe that for $G = G_1 \otimes 0^- \in \Lb_{\alpha^+}\otimes 0^-$ we have
 $(B'G)(\eta) = 0^{|\eta^-|}(\overline{B}G_1)(\eta^+)$, i.e. $B' P_+ = P_+ \overline{B}$. Take $n \geq 1$ and let $\alpha_j^{\pm}$ be given as before. Then
 \begin{align*}
  &\ \Vert H_n^{\e}(t)P_+G_1 - P_+\overline{H}_n(t)G \Vert_{\Lb_{\alpha}}
  \leq \int \limits_{0}^{t} \Vert T_{\e}(t-s)B'\left( H_{n-1}^{\e}(s)P_+G_1 - P_+\overline{H}_{n-1}(s)G_1 \right) \Vert_{\Lb_{\alpha}}ds
  \\ &\ \ \ + \int \limits_{0}^{t} \Vert \left(T_{\e}(t-s)P_+ - P_+\overline{T}(t-s)\right) \overline{B}\overline{H}_{n-1}(s)G_1 \Vert_{\Lb_{\alpha}}ds
  \\ &\leq \Vert B' \Vert_{L(\Lb_{\alpha_1}, \Lb_{\alpha_0})} \int \limits_{0}^{t} \Vert  H_{n-1}^{\e}(s)P_+G_1 - P_+\overline{H}_{n-1}(s)G_1  \Vert_{\Lb_{\alpha_1}}ds
  + T a^{\e}_{n-1}
 \end{align*}
 where $K^{(n-1)}_{\alpha^+} = \{ \overline{B}\overline{H}_{n-1}(s)G \ | \ 0 \leq s \leq T \} \subset \Lb_{\alpha^+}$ is compact and
 \[
  a^{\e}_{n-1} = \sup \limits_{(t,F) \in [0,T] \times K^{(n-1)}_{\alpha^+}} \Vert T_{\e}(t) P_+F - P_+\overline{T}(t)F \Vert_{\Lb_{\alpha}}.
 \]
 By \eqref{EQ:40} we see that $a_{n-1}^{\e} \longrightarrow 0$ as $\e \to 0$. Iteration yields
 \begin{align*}
   &\ \Vert H_n^{\e}(t)P_+G_1 - P_+\overline{H}_n(t)G \Vert_{\Lb_{\alpha}}
   \\ &\leq C_1^{(n)} \int \limits_{0}^{t}\cdots \int \limits_{0}^{t_{n-1}} \Vert H_0^{\e}(s)P_+G_1 - P_+\overline{H}_0(s)G_1 \Vert_{\Lb_{\beta}}dt_{n}\cdots dt_1
   + C(T,n,a_0^{\e}, \dots, a_{n-1}^{\e})
 \end{align*}
 where $C_1^{(n)}$ depends on $\Vert B' \Vert_{L(\Lb_{\alpha_{j+1}}, \Lb_{\alpha_j})}$, $j= 0, \dots, n-1$ and 
 $C(T,n,a_0^{\e}, \dots, a_{n-1}^{\e}) \longrightarrow 0$ as $\e \to 0$. 
 The iterated integral tends by \eqref{EQ:40} uniformly in $t \in [0,T]$ to zero which proves \eqref{EQ:43}.
\end{proof}

\subsection{Proof of Theorem \ref{TH:05}}
Let $F \in \mathcal{FP}(\Gamma^+)$ and take $G \in B_{bs}(\Gamma_0^+)$ such that $KG = F$. Then
\begin{align*}
 \int \limits_{\Gamma^2}F(\gamma^+)d \mu_t^{\e}(\gamma) &= \int \limits_{\Gamma_0^+}G(\eta^+)k_{\mu_t^{\e}}(\eta^+,\emptyset)d\lambda(\eta^+),
 \\ \int \limits_{\Gamma^+}F(\gamma^+)d\overline{\mu}_t(\gamma^+) &= \int \limits_{\Gamma_0^+}G(\eta^+)k_{\overline{\mu}_t}(\eta^+)d\lambda(\eta^+).
\end{align*}
Hence it suffices to prove for each $G \in B_{bs}(\Gamma_0^+)$
\begin{align}\label{EQ:44}
 \int \limits_{\Gamma_0^+}G(\eta^+)k_{\mu_t^{\e}}(\eta^+,\emptyset)d\lambda(\eta^+) \longrightarrow \int \limits_{\Gamma_0^+}G(\eta^+)k_{\overline{\mu}_t}(\eta^+)d\lambda(\eta^+)
\end{align}
as $\e \to 0$ uniformly in $t \in [0, T]$. 
This will be proved by arguments given in the proof of Theorem \ref{TH:03}.
Namely, let $k_{t,\e}^{(n)} \in \K_{\alpha_n}$, $t \in [0, T^{(n)}]$ be given by \eqref{EQ:34} with $T^{(n)}$ given by \eqref{EQ:04}.
Moreover, $\alpha_n$ is given by $\alpha_n^+ = \alpha^+ + (\lambda + \| a^+ \|_{L^1})S_n$ and $\alpha_n^- = \alpha^-$.
Similarly let $r_t^{(1)} := \overline{U}^{\Delta}(t)k_{\mu_0^+}$ and
\[
 r_t^{(n+1)} := \overline{U}^{\Delta}(t)r_{T^{(n)}}^{(n)}, \ \ t \in [0,T^{(n+1)}].
\]
\begin{Lemma}
 For each $n \geq 1$ it holds that
 \[
  \sup \limits_{t \in [0,T^{(n)}]}\ \left| \int \limits_{\Gamma_0^+}G(\eta^+)k_{t,\e}^{(n)}(\eta^+,\emptyset)d\lambda(\eta^+) -  \int \limits_{\Gamma_0^+}G(\eta^+)r_t^{(n)}(\eta^+)d\lambda(\eta^+)\right| \longrightarrow 0, \ \ \e \to 0.
 \] 
\end{Lemma}
\begin{proof}
 Observe that $G \in B_{bs}(\Gamma_0^2) \subset \Lb_{\alpha_n}$ for all $n \geq 1$. Hence all integrals given below are well-defined.
 For $n = 1$ we obtain
 \begin{align*}
  \int \limits_{\Gamma_0^+}G(\eta^+)k_{t,\e}^{(1)}(\eta^+,\emptyset)d\lambda(\eta^+)
  = \int \limits_{\Gamma_0^2}(P_+G)(\eta)U_{\e}^{\Delta}(t)k_{\mu_0}(\eta)d \lambda(\eta)
  = \int \limits_{\Gamma_0^2} (\widehat{U}_{\e}(t)P_+G)(\eta)k_{\mu_0}(\eta)d \lambda(\eta)
 \end{align*}
 and from \eqref{EQ:45} it follows that
 \[
  \sup \limits_{t \in [0,T^{(1)}]}\left| \int \limits_{\Gamma_0^2} (\widehat{U}_{\e}(t)P_+G)(\eta)k_{\mu_0}(\eta)d \lambda(\eta) - \int \limits_{\Gamma_0^2}(P_+\overline{U}(t)G)(\eta)k_{\mu_0}(\eta)d \lambda(\eta)\right| \longrightarrow 0, \ \ \e \to 0.
 \]
 Moreover we have by $k_{\mu_0}(\eta^+,\emptyset) = k_{\mu_0^+}(\eta^+)$
 \begin{align*}
  \int \limits_{\Gamma_0^2}(P_+\overline{U}(t)G)(\eta)k_{\mu_0}(\eta)d \lambda(\eta)
  &= \int \limits_{\Gamma_0^+}(\overline{U}(t)G)(\eta^+)k_{\mu_0^+}(\eta^+)d \lambda(\eta^+)
  \\ &= \int \limits_{\Gamma_0^+} G(\eta^+) (\overline{U}^{\Delta}(t)k_{\mu_0^+})(\eta^+)d \lambda(\eta^+)
 \end{align*}
 which shows the assertion for $n = 1$. Proceeding by induction we obtain
 \begin{align*}
  \int \limits_{\Gamma_0^+}G(\eta^+)k_{t,\e}^{(n)}(\eta^+,\emptyset)d\lambda(\eta^+)
  &= \int \limits_{\Gamma_0^2}(P_+G)(\eta) (U_{\e}^{\Delta}(t)k_{T^{(n-1)},\e}^{(n-1)})(\eta)d \lambda(\eta)
  \\ &= \int \limits_{\Gamma_0^2} (\widehat{U}_{\e}(t)P_+G)(\eta) k_{T^{(n-1)},\e}^{(n-1)}(\eta)d \lambda(\eta).
 \end{align*}
 Then
 \begin{align*}
  &\ \left| \langle \langle \widehat{U}_{\e}(t)P_+G, k_{T^{(n-1)},\e}^{(n-1)} \rangle \rangle - \langle \langle P_+\overline{U}(t)G, r_t^{(n-1)} \rangle \rangle \right|
  \\ &\leq \left| \langle \langle \widehat{U}_{\e}(t)P_+G,  k_{T^{(n-1)},\e}^{(n-1)} \rangle \rangle - \langle \langle \widehat{U}_{\e}(t)P_+G, r_t^{(n-1)}\rangle \rangle \right|
       + \left| \langle \langle \widehat{U}_{\e}(t)P_+G, r_t^{(n-1)} \rangle - \langle \langle P_+\overline{U}(t)G, r_t^{(n-1)} \rangle \rangle \right|
  \\ &= I_1 + I_2.
 \end{align*}
 For the first term we get
 $I_1 \leq \sup \limits_{F \in K_{\beta}} \left| \langle F, k_{T^{(n-1)},\e}^{(n-1)} \rangle - \langle F, r_t^{(n-1)} \rangle \right|$
 where by $G \in B_{bs}(\Gamma_0^2) \subset \Lb_{\alpha_{n-1}}$
 \[
  K_{n-1} = \{ \widehat{U}_{\e}(t)P_+G \ | \ t \in [0,T^{(n)}], \ \e \in (0,1] \} \cup \{ P_+\overline{U}(t)G \ | \ t \in [0,T^{(n)}] \} \subset \Lb_{\alpha_{n-1}}
 \]
 is compact, see \eqref{EQ:45}. By induction hypothesis we can show that $I_1 \longrightarrow 0$ uniformly in $[0,T^{(n)}]$.
 Concerning $I_2$ we obtain
 \[
  I_2 \leq \sup \limits_{t \in [0,T^{(n)}]}\Vert r_t^{(n-1)} \Vert_{\K_{\alpha_{n-1}}} \Vert \widehat{U}_{\e}(t)P_+G - P_+\overline{U}(t)G\Vert_{\Lb_{\alpha_{n-1}}}
 \]
 and the assertion follows by \eqref{EQ:45}.
\end{proof}
In view of \eqref{EQ:46}, \eqref{CORRDEF} and similarly
\[
 k_{\overline{\mu}_t} = r_{t - S_{n-1}}^{(n)}, \ \ S_{n-1} \leq t \leq S_n, \ \ n \geq 1
\]
we see that \eqref{EQ:44} holds. This completes the proof of Theorem \ref{TH:05}.

\subsubsection*{Acknowledgments}
...

\section{Appendix}

\begin{Lemma}\label{APPENDIX:00}
 Let $X$ be a Banach space and $(T_n)_{n \in \N} \subset L(X)$ be such that $T_n \longrightarrow T$ strongly.
 Take a compact $K \subset X$. Then 
 \[
  \sup \limits_{x \in K}\Vert T_nx - Tx \Vert_{X} \longrightarrow 0, \ \ n \to \infty.
 \]
\end{Lemma}

\begin{Lemma}\label{APPENDIX:01}
 Let $X$ be a Banach space and $X^*$ its dual. Let $(x_n^*)_{n \in \N} \subset X^*$ be such that 
 \[
  x_n^*(x) \longrightarrow x^*(x), \ \ n \to \infty
 \]
 holds for all $x \in X$. Then for any compact $K \subset X$
 \[
  \sup \limits_{x \in K}| x_n^*(x) - x^*(x)| \longrightarrow 0, \ \ n \to \infty.
 \]
\end{Lemma}

\newpage

\begin{footnotesize}

\bibliographystyle{alpha}
\bibliography{Bibliography}

\end{footnotesize}

\end{document}